\documentclass[a4paper,11pt,USenglish,fleqn]{article}
\RequirePackage[T1]{fontenc}
\RequirePackage{lmodern}

\makeatletter
\renewcommand\section{\@startsection {section}{1}{\z@}%
	{-3.5ex \@plus -1ex \@minus -.2ex}%
	{2.3ex \@plus.2ex}%
	{\sffamily\Large\bfseries\raggedright}}
\renewcommand\subsection{\@startsection{subsection}{2}{\z@}%
	{-3.25ex\@plus -1ex \@minus -.2ex}%
	{1.5ex \@plus .2ex}%
	{\sffamily\Large\bfseries\raggedright}}
\renewcommand\subsubsection{\@startsection{subsubsection}{3}{\z@}%
	{-3.25ex\@plus -1ex \@minus -.2ex}%
	{1.5ex \@plus .2ex}%
	{\sffamily\Large\bfseries\raggedright}}
\renewcommand\paragraph{\@startsection{paragraph}{4}{\z@}%
	{-3.25ex \@plus-1ex \@minus-.2ex}%
	{1.5ex \@plus .2ex}%
	{\sffamily\large\bfseries\raggedright}}
\renewcommand\subparagraph{\@startsection{subparagraph}{5}{\z@}%
	{3.25ex \@plus1ex \@minus .2ex}%
	{-1em}%
	{\sffamily\normalsize\bfseries}}
\makeatother

\sffamily
\usepackage{helvet}
\sffamily
\DeclareFontShape{T1}{lmss}{bx}{sc} { <-> ssub * phv/bx/sc }{}
\DeclareFontShape{T1}{lmss}{m}{sc} { <-> ssub * phv/m/sc }{}

\makeatletter
\g@addto@macro\bfseries{\boldmath}
\makeatother

\RequirePackage[utf8]{inputenc}
\usepackage{microtype}

\usepackage[top=2cm,bottom=3cm,left=2cm,right=2cm]{geometry}
\usepackage[amsthm]{ntheorem}
\usepackage{booktabs}
\usepackage{multirow}
\usepackage[labelsep=period,labelfont={footnotesize,bf},font={footnotesize},skip=2mm]{caption}
\usepackage{subcaption}
\usepackage{xspace}
\usepackage{tikz}
\usepackage{algorithm}
\usepackage{algorithmicx}
\usepackage[noend]{algpseudocode}
\usepackage{refcount}
\usepackage{amssymb}
\usepackage{amsmath}
\usepackage{mathtools}
\usepackage{dsfont}
\usepackage[colorlinks,linkcolor=blue,citecolor=blue,urlcolor=blue,hypertexnames=false]{hyperref}
\usepackage[nameinlink]{cleveref}

\theoremstyle{plain}
\newtheorem{theorem}{Theorem}
\newtheorem{lemma}[theorem]{Lemma}
\newtheorem{corollary}[theorem]{Corollary}
\newtheorem{definition}[theorem]{Definition}

\newtheorem{theoremduplicate}{Theorem}
\newtheorem{lemmaduplicate}{Lemma}

\algrenewcommand\algorithmicindent{2.4em}
\renewcommand{\setminus}{-}
\newcommand{\add}{\cup}
\newcommand{\e}{\vspace*{-1em}}
\newcommand{\clause}{C}
\newcommand{\formula}{\Phi}
\allowdisplaybreaks

\newcommand{\N}{\ensuremath{\mathds{N}}}
\newcommand{\CNF}{\ensuremath{\textsc{CNF}}\xspace}
\newcommand{\ThreeCNF}{\ensuremath{\textsc{3}\textsc{-}\textsc{CNF}}\xspace}
\newcommand{\EThreeCNF}{\ensuremath{\textsc{3}\textsc{-}\textsc{CNF}}\xspace}
\newcommand{\TwoOrThreeCNF}{\ensuremath{\textsc{2or3}\textsc{-}\textsc{CNF}}\xspace}

\newcommand{\ThreeOccTwoOrThreeCNF}{\ensuremath{\textsc{3Occurrences}\textsc{-}\textsc{2or3}\textsc{-}\textsc{CNF}}\xspace}

\newcommand{\Sat}{\ensuremath{\textsc{Sat}}\xspace}
\newcommand{\ThreeSat}{\ensuremath{\textsc{3}\textsc{-}\textsc{Sat}}\xspace}
\newcommand{\EThreeSat}{\ensuremath{\textsc{3}\textsc{-}\textsc{Sat}}\xspace}

\newcommand{\ThreeUnSat}{\ensuremath{3\textsc{-}\textsc{UnSat}}\xspace}
\newcommand{\MinimalUnSat}{\ensuremath{\textsc{Minimal}\textsc{-}\allowbreak{}\textsc{Un}\-\textsc{Sat}}\xspace} 
\newcommand{\MinimalThreeUnSat}{\ensuremath{\textsc{Minimal}\textsc{-}\allowbreak{}3\textsc{-}\textsc{UnSat}}\xspace} 

\newcommand{\DecVC}{\ensuremath{\textsc{Vertex}\-\textsc{Cover}}\xspace}
\newcommand{\ThreeCol}{\ensuremath{\textsc{3}\textsc{-}\textsc{Color}\-\textsc{ability}}\xspace}
\newcommand{\DecCol}{\ensuremath{\textsc{Color}\-\textsc{ability}}\xspace}

\newcommand{\MinimalNonVC}{\ensuremath{\textsc{Minimal}\textsc{-}\allowbreak{}k\textsc{-}\textsc{No}\-\textsc{Vertex}\-\textsc{Cover}}\xspace}

\newcommand{\VertexMinimalThreeUnCol}{\ensuremath{\textsc{Vertex}\-\textsc{Minimal}\textsc{-}\allowbreak{}3\textsc{-}\textsc{Un}\-\textsc{Color}\-\textsc{ability}}\xspace} 
\newcommand{\EdgeMinimalThreeUnCol}{\ensuremath{\textsc{Minimal}\textsc{-}\allowbreak{}3\textsc{-}\textsc{Un}\-\textsc{Color}\-\textsc{ability}}\xspace} 
\newcommand{\VertexMinimalKayUnCol}{\ensuremath{\textsc{Vertex}\-\textsc{Minimal}\textsc{-}\allowbreak{}k\textsc{-}\textsc{Un}\-\textsc{Color}\-\textsc{ability}}\xspace} 
\newcommand{\EdgeMinimalKayUnCol}{\ensuremath{\textsc{}\-\textsc{Minimal}\textsc{-}\allowbreak{}k\textsc{-}\textsc{Un}\-\textsc{Color}\-\textsc{ability}}\xspace}

\newcommand{\OptProb}{\ensuremath{\textsc{Opt}\-\textsc{Prob}}\xspace} 
\newcommand{\opt}{\mbox{\textnormal{opt}}}
\newcommand{\sat}{\ensuremath{s}}
\newcommand{\C}{\ensuremath{\textnormal{C}}\xspace}
\newcommand{\F}{\ensuremath{\textnormal{F}}\xspace}
\newcommand{\T}{\ensuremath{\textnormal{T}}\xspace}
\renewcommand{\P}{\ensuremath{\textnormal{P}}\xspace}
\newcommand{\NP}{\ensuremath{\textnormal{NP}}\xspace}
\newcommand{\coNP}{\ensuremath{\textnormal{coNP}}\xspace}
\newcommand{\DP}{\ensuremath{\textnormal{DP}}\xspace}

\newcommand{\ihat}{{\hat\imath}}
\newcommand{\jhat}{{\hat\jmath}}
\newcommand{\alg}{\ensuremath{A}\xspace}

\crefname{section}{Section}{Sections}
\crefname{subsection}{Subsection}{Subsections}
\crefname{lemma}{Lemma}{Lemmas}
\crefname{figure}{Figure}{Figures}
\crefname{table}{Table}{Tables}
\crefname{theorem}{Theorem}{Theorems}
\crefname{definition}{Definition}{Definitions}
\crefname{corollary}{Corollary}{Corollaries}
\crefname{equation}{Equation}{Equations}
\crefname{algorithm}{Algorithm}{Algorithms}
\crefname{appendix}{Appendix}{Appendices}

\let\leftold\left
\let\rightold\right
\renewcommand{\left}{\mathopen{}\mathclose\bgroup\leftold}
\renewcommand{\right}{\aftergroup\egroup\rightold}

\title{\bfseries Finding Optimal Solutions With Neighborly Help}

\author{
  Elisabet Burjons\\
  \small Department of Computer Science, ETH Zürich\\
  \small \url{eburjons@inf.ethz.ch}\\[2mm]
  Fabian Frei\\
  \small Department of Computer Science, ETH Zürich\\
  \small \url{fabian.frei@inf.ethz.ch}\\[2mm]
  Edith Hemaspaandra\\
  \small Department of Computer Science, Rochester Institute of Technology\\
  \small \url{eh@cs.rit.edu}\\[2mm]
  Dennis Komm \\
  \small Department of Computer Science, ETH Zürich\\
  \small \url{dennis.komm@inf.ethz.ch}\\[2mm]
  David Wehner \\
  \small Department of Computer Science, ETH Zürich\\
  \small \url{david.wehner@inf.ethz.ch}}

\date{}

\begin{document}

\maketitle

\begin{quote}\small
  \textbf{Abstract.}
  Can we efficiently compute optimal solutions to instances of a hard problem from optimal
  solutions to neighboring (i.e., locally modified) instances? 
  For example, can we efficiently compute an optimal coloring for a graph 
  from optimal colorings for all one-edge-deleted subgraphs?
  Studying such questions not only gives detailed insight into the structure of
  the problem itself, but also into the complexity of related problems;
  most notably graph theory's core notion of critical graphs 
  (e.g., graphs whose chromatic number decreases under deletion of an arbitrary edge)
  and the complexity-theoretic notion of minimality problems  
  (also called criticality problems, e.g., recognizing graphs that become $3$-colorable when an arbitrary edge is deleted).
  
  We focus on two prototypical graph problems, 
  Colorability and Vertex Cover. 
  For example, we show that it is \NP-hard 
  to compute an optimal coloring for a graph from optimal colorings 
  for \emph{all} its one-vertex-deleted subgraphs, 
  and that this remains true even when  optimal solutions for \emph{all} one-edge-deleted subgraphs are given. 
  In contrast, computing an optimal coloring from all (or even just two) one-edge-added
  supergraphs is in \P. 
  We observe that Vertex Cover exhibits a remarkably different behavior, demonstrating the power of our model to delineate problems from each other more precisely on a structural level.
  
  Moreover, we provide a number of new complexity results for minimality and criticality problems. For example, we prove that \EdgeMinimalThreeUnCol 
  is complete for \DP (differences of \NP sets),
  which was previously known only for the more amenable 
  case of deleting vertices rather than edges.
  For Vertex Cover, we show that recognizing $\beta$-vertex-critical graphs 
  is complete for $\Theta_2^\textnormal{p}$
  (parallel access to \NP), obtaining the first completeness result 
  for a criticality problem for this class. 
\end{quote}

\begin{quote}\small
\end{quote}

\begin{quote}\small
	\textbf{Funding.}
	\emph{Edith Hemaspaandra:} Research done in part while on sabbatical at ETH Zürich.
\end{quote}
\newpage
\section{Introduction and Related Work}\label{sec:introduction}

In \cref{sec:model}, we introduce and motivate our new model, which we then compare and contrast to related notions in \cref{sec:related-concepts}. Finally, we present in \cref{sec:results} an overview of our most interesting results and place them into the context of the wider literature. 

\subsection{Our Model}\label{sec:model}

In view of the almost complete absence of progress in the question of $\P$ versus $\NP$,  
it is natural to wonder just how far and in what way these sets may differ. 
For example, how much additional information enables us to design an algorithm that solves an otherwise \NP-hard problem in polynomial time? 
We are specifically interested in the case where this additional information takes the form of 
optimal solutions to neighboring (i.e., locally modified) instances. 
This models situations such as that of 
a newcomer 
who 
may ask 
experienced 
peers for advice on how to solve a difficult problem, 
for instance finding an optimal work route. 
Similar circumstances arise when new servers join a computer network. 
Formally, we consider the following oracle model:
An algorithm may, on any given input, repeatedly select an arbitrary  instance neighboring the given one 
and query the oracle for an optimal solution to it.
Occasionally, it will be interesting to limit the number of queries that we grant the algorithm. In general, we do not impose such a restriction, however. 

What precisely constitutes a local modification and thus a neighbor depends on the specific problem, of course. 
We examine the prototypical graph problems Colorability and Vertex Cover, considering the following four local modifications, which are arguably the most natural choices: 
deleting an edge,
adding an edge,
deleting a vertex (including adjacent edges), and 
adding a vertex (including an arbitrary, possibly empty, set of edges from the added vertex to the existing ones). 
For example, we ask whether there is a polynomial-time algorithm that computes a minimum vertex cover for an input graph $G$ if it has access to minimum vertex covers for all one-edge-deleted subgraphs of $G$. 
We will show that questions of this sort are closely connected to and yet clearly distinct from 
research in other areas, in particular the study of critical graphs, minimality problems, self-reducibility, and reoptimization. 

\subsection{Related Concepts}\label{sec:related-concepts} 
\subparagraph*{Criticality.}
The notion of criticality was introduced into the field of graph theory by Dirac~\cite{dir:theorems-abstract-graphs} 
in 1952
in the context of Colorability with respect to vertex deletion.
Thirty years later, Wessel~\cite{wes:criticality-graph-theory} 
generalized the concept to arbitrary graph properties 
and modification operations. 
Nevertheless, Colorability has remained a central focus of the extensive research on critical graphs.
Indeed, a graph $G$ is called \emph{critical} without any further specification if it is $\chi$-critical under edge deletion, that is, if its chromatic number $\chi(G)$ 
(the number of colors used in 
an optimal coloring of $G$) 
changes when an arbitrary edge is deleted.
Besides Colorability, one other problem has received a comparable amount of attention and thorough analysis in 
three different manifestations: Independent Set, Vertex Cover, and Clique.
The corresponding notions are $\alpha$-criticality, $\beta$-criticality, and $\omega$-criticality, 
where $\alpha$ is the \emph{independence number} (size of a maximum independent set), 
$\beta$ is the \emph{vertex cover number} (size of a minimum vertex cover), and 
$\omega$ is the \emph{clique number} (size of a maximum clique). 
Note that these graph numbers are all monotone---either nondecreasing or nonincreasing---with respect to each of the local modifications examined in this paper.\e

\subparagraph*{Minimality.}
Another strongly related notion is that of minimality problems. 
An instance is called \emph{minimal} with respect to a property if 
only the instance itself but none of its neighbors has this property; 
that is, it inevitably loses the property under the considered local modification.
The corresponding minimality problem is to decide whether an instance is minimal in the described sense. 
For example, a graph $G$ is minimally 3-uncolorable (with respect to edge deletion) if it is not 3-colorable, yet all its one-edge-deleted neighbors are. The minimality problem \EdgeMinimalThreeUnCol is the set of all minimally 3-uncolorable graphs. 
Note that a graph is critical exactly if it is minimally $k$-uncolorable for some $k$. 

While minimality problems tend to be in \DP 
(i.e., differences of two \NP sets, the second level of the Boolean hierarchy),
\DP-hardness is so difficult to prove for them that only a few have been shown 
to be \DP-complete so far; see for instance Papadimitriou and Wolfe~\cite{pap-wol:j:facets}.
Note that the notion of minimality is not restricted to graph problems. Indeed, minimally unsatisfiable formulas
figure prominently
in many of our proofs.\e

\subparagraph*{Auto-Reducibility.}
Our model provides a refinement of the notion of functional 
auto-reducibility; see Faliszewski and Ogihara~\cite{fal-ogi:j:func-autored}.
An algorithm solves a function problem $R\subseteq\Sigma_1^\ast\times\Sigma_2^\ast$ 
if on input $x\in\Sigma_1^\ast$ it outputs some  $y\in\Sigma_2^\ast$ with $(x,y)\in R$. 
The problem $R$ is \emph{auto-reducible} if there is a polynomial-time algorithm 
with unrestricted access to an oracle that provides solutions to all instances except $x$ itself.
The task of finding an optimal solution to a given instance 
is a special kind of function problem. 
Defining all instances to be neighbors (local modifications) of each other 
lets the two concepts coincide.\e

\subparagraph*{Self-Reducibility.}
Self-reducibility is auto-reducibility with the additional restriction 
that the algorithm may query the oracle only on instances that are smaller in a certain way. 
There are a multitude of definitions of self-reducibility that differ in what exactly is considered to be ``smaller,'' 
the two seminal ones stemming from Schnorr~\cite{sch:c:self-reducible}
and from Meyer and Paterson~\cite{mey-pat:t:int}. 
For Schnorr, an instance is smaller than another one if its encoding input string is strictly shorter. 
While his definition does allow for functional problems (i.e., more than mere decision problems, in particular the problem of finding an optimal solution), 
it is too restrictive for self-reducibility to encompass our model since not all neighboring graphs have shorter strings under natural encodings.
Meyer and Paterson are less rigid and allow instead any partial order having short downward chains to determine which instances are considered smaller than the given one.\footnote{Formally, a partial order is said to \emph{have short downward chains} if the following condition is satisfied: There is a polynomial $p$ such that every chain decreasing 
with respect to the considered partial order and starting with some string $x$ is shorter than 
$p(|x|)$ and such that all strings 
preceding $x$ in that order are bounded in length by $p(|x|)$.}
The partial orders induced by deleting vertices, by deleting edges, and by adding edges all have short downward chains. The definition by Meyer and Paterson is thus sufficient for 
our model to become part of functional self-reducibility 
for all local modifications considered in this paper but one, namely, the case of adding a vertex, which is too generous a modification to display any particularly interesting behavior.

As an example, consider the graph decision problem
$\DecCol=\{(G,k)\mid\chi(G)\le k\}$, which is self-reducible 
by the following
observation. 
Any graph $G$ with at least two vertices that is not a clique is $k$-colorable
exactly if at least one of the polynomially many graphs that result from 
merging two non-adjacent vertices in $G$ is $k$-colorable. 
This works for the optimization variant of the problem as well. 
Any optimal coloring of $G$ assigns at least two vertices the same
color, 
except in the trivial case of $G$ being a clique. An optimal coloring for the graph 
that has 
two such vertices merged then  
yields an optimal coloring for $G$. 
This 
contrasts well with 
the findings 
for Colorability's behavior under our new model discussed 
below.\e

\subparagraph*{Reoptimization.}
Reoptimization examines optimization problems under a model that is tightly connected to ours.
The notion of reoptimization was coined by Schäffter \cite{sch:forbidden} 
and first applied by Archetti et al.~\cite{arc-ber-spe:reoptimizing-tsp}.
The reoptimization model sets the following task for an optimization problem:
\begin{quote}
\hspace*{-0.4ex}\emph{Given an instance, an optimal solution to it, and a local modification of this instance, compute an optimal solution to the modified instance.}
\end{quote}
The proximity to our model becomes clearer after a change of perspective. 
We reformulate the reoptimization task by 
reversing the roles of the given and the modified instance.
\begin{quote}
\hspace*{-0.4ex}\emph{Given an instance, a local modification of it, and an optimal solution to the modified instance, compute an optimal solution to the original instance.}
\end{quote}
Note that this perspective switch flips the definition of local modification; for example, edge deletion turns into edge addition.
Aside from this, the task now reads almost identical to that demanded in our model. 
The sole but crucial difference is that in reoptimization, the neighboring instance and the optimal solution to it are given as part of the input, whereas in our model, the algorithm may select any number of neighboring instances and query the oracle for optimal solutions to them.
Even if we limit the number of queries to just one, our model is still more generous since the algorithm is choosing (instead of being given) the neighboring instance to which the oracle will supply an optimal solution.
Thus, hardness in our model always implies hardness for reoptimization, but not vice versa.
In fact, all problems examined under the reoptimization model so far remain NP-hard. Only for some of them could an improvement of the approximation ratio be achieved after extensive studies, the first discovered examples being \textsc{tsp} under edge-weight changes \cite{boe-for-etal:reusing-optimal-solutions} and addition or deletion of vertices \cite{aus-esc-etal:reoptimization-tsp}. This stands in stark contrast to the results for our model, as outlined in the next section.

\subsection{Results}\label{sec:results}

We shed a new light on two of the most prominent and well-examined graph problems, 
Colorability and Vertex Cover.

Our results come in two different types. 

The first type concerns the hardness of the two problems 
in our model for the most common local modifications; 
\cref{tab:resultoverview} summarizes the main results of this type.
\begin{table}
  \begin{center}
    \caption{An overview of our results regarding the hardness of Colorability and Vertex Cover in our model for the most common definitions of a local modification. The $v$ stands for a vertex and the $e$ stands for an edge. The question mark indicates an interesting open problem. The results in the vertex-addition columns are trivial; see \cref{thm:addvertextheorem} in \cref{app:addvertextheorem}. The NP-hardness results for the 1-query case all follow from rather simple Turing reductions; see \cref{thm:onequerylemma} in \cref{app:onequerylemma}.}
    \label{tab:resultoverview}
    \setlength{\tabcolsep}{3pt}\footnotesize
    \begin{tabular}{ccccccccc}
      \toprule
      \multirow{2}{*}[-1ex]
      {\shortstack[1]{No.\ of\\Queries}} & \multicolumn{4}{c}{Colorability} 
      & \multicolumn{4}{c}{Vertex Cover} \\ \cmidrule(lr){2-5}\cmidrule(lr){6-9} 
      & Add $v$ 
      & Delete $v$ 
      &Add $e$
      & Delete $e$ 
      & Add $v$ 
      & Delete $v$ 
      &Add $e$
      & Delete $e$ \\ \midrule
      1  
      & $\P$
      & $\NP\text{-hard}$
      & $\NP\text{-hard}$
      & $\NP\text{-hard}$   
      & $\P$
      & $\NP\text{-hard}$
      & $\NP\text{-hard}$
      & $\NP\text{-hard}$\\ \midrule
      \multirow{2}{*}[0ex]{\shortstack[1]{2 or\\more}}    
      &$\P$
      & $\NP\text{-hard}$
      & $\P$
      &$\NP\text{-hard}$
      & $\P$
      & $\P$
      & $\P$
      &?\\
      &[Thm.~\ref{thm:addvertextheorem}]
      &[Thm.~\ref{thm:color-v}]
      &[Thm.~\ref{thm:colorer}]
      &[Thm.~\ref{thm:color-e}]
      &[Thm.~\ref{thm:addvertextheorem}]
      &[Thm.~\ref{thm:removevPoly}]
      &[Thm.~\ref{thm:addePoly}]
      &\\\bottomrule
    \end{tabular}
  \end{center}
\end{table}
In addition, \cref{cor:SATisNPhardInOurModel,cor:removeTrianglePoly} 
show that Satisfiability and Vertex Cover remain \NP-hard for any number of queries if the local modification is the deletion of a clause or a triangle, respectively. 
The results for the vertex-addition columns are trivial since we can just query an optimal solution for the graph with an added isolated vertex; see \cref{thm:addvertextheorem}. 
The hardness results for the one-query case all follow from the same simple \cref{thm:onequerylemma}, variations of which appear in the study of self-reducibility and many other fields; see \cref{app:onequerylemma}.
The findings of \cref{thm:colorer,thm:removevPoly,thm:addePoly} clearly delineate our model from that in reoptimization, where the \NP-hard problems examined in the literature remain \NP-hard despite the significant amount of advice in form of the provided optimal solution; see Böckenhauer et al.~\cite{handbook-reoptimization}.

The results of the second type locate criticality problems in relation to the complexity classes \DP and $\Theta_2^\textnormal{p}$. 
The class $\Theta_2^\textnormal{p}$ was introduced by Wagner~\cite{wag:j:bounded} and represents the languages that can be decided in polynomial time by an algorithm that has access to an \NP oracle under the restriction that all queries are submitted at the same time. The definitions of the classes immediately yield the inclusions $\NP\cup\coNP\subseteq \DP\subseteq\Theta_2^\textnormal{p}$.

Papadimitriou and Wolfe~\cite{pap-wol:j:facets} have shown that \MinimalUnSat 
(the set of unsatisfiable formulas that become satisfiable when an arbitrary clause is deleted) is \DP-complete.
Cai and Meyer~\cite{cai-mey:j:dp} built upon this to prove \DP-completeness of 
\VertexMinimalKayUnCol (the
set of graphs that are not $k$-colorable but become $k$-colorable 
when an arbitrary vertex is deleted),
for all $k \geq 3$. 
With \cref{thm:MU-DP-complete,thm:EdgeMinimalKayUnCol}, we
were able to extend this result to classes that are analogously defined for the much smaller local modification of edge deletion, which is considered the default setting; 
namely, we prove \DP-completeness of 
\EdgeMinimalKayUnCol, for all $k \geq 3$.

In \cref{thm:criticalityandstability}, we show that recognizing criticality and vertex-criticality are in $\Theta_2^\textnormal{p}$ and \DP-hard.
As Joret~\cite
{jor:thesis:entropy-and-stability} points out, 
a construction by 
Papadimitriou and Wolfe~\cite{pap-wol:j:facets} 
proves the \DP-hardness of recognizing $\beta$-critical graphs. 
This problem also lies in $\Theta_2^\textnormal{p}$, 
but no finer classification has been achieved so far.
In \cref{thm:VertexCriticalThetaComplete}, we show that this problem is in fact $\Theta_2^\textnormal{p}$-hard, 
yielding the first known 
$\Theta_2^\textnormal{p}$-completeness result for a criticality problem.

\section{Preprocessing \EThreeSat}\label{sec:sat}

Our main technique for proving the nontrivial hardness results in our model is the following: 
We build in polynomial-time computable solutions
for each locally modified problem instance. That way, the solutions
to the locally modified problem instances do not give away any information
about the instance to be solved. A similar approach is taken in some proofs of
\DP-completeness for minimality 
problems.
Indeed, we can occasionally combine the proof of
\DP-hardness with that of the \NP-hardness of
computing an optimal solution from optimal solutions to locally
modified instances.
Denote by \EThreeCNF the set of nonempty \CNF-formulas with exactly three distinct literals per clause.\footnote{This set is often denoted $\textsc{E3-CNF}$ in the literature.} 
We begin by showing in \cref{thm:3sat} that there is a reduction 
from \EThreeSat (the set of satisfiable \EThreeCNF-formulas) to \EThreeSat 
that builds in polynomial-time computable solutions for
all one-clause-deleted subformulas of the resulting \EThreeCNF-formula. At first glance, this very surprising result may seem dangerously close to proving $\P=\NP$; \cref{cor:SATisNPhardInOurModel} will make explicit where the hardness remains.
We will then use the reduction of \cref{thm:3sat} 
as a preprocessing step in reductions from \EThreeSat to other problems.

\begin{theorem}\label{thm:3sat}
  There is a polynomial-time many-one reduction $f$ from \EThreeSat to \EThreeSat
  and a polynomial-time computable function $\sat$ 
  such that, for every \EThreeCNF-formula $\formula$ and
  for every clause ${\clause}$ in $f(\formula)$,
  $\sat(f(\formula) - {\clause})$ is a satisfying assignment for $f(\formula) - {\clause}$.
\end{theorem}

\begin{proof}
  Papadimitriou and Wolfe~\cite[Lemma 1]{pap-wol:j:facets}
  give a reduction from \ThreeUnSat to \MinimalUnSat (the set of \CNF-formulas
  that are unsatisfiable but that become satisfiable with the removal of an arbitrary clause). 
  In \cref{sec:satselfreduction}, we show how to enhance this reduction such that it has all properties of our theorem. 
  First, we carefully prove that there is a function 
  $\sat$ 
  that together with the original reduction satisfies all properties of our theorem, except that we may output a formula that is not in \ThreeCNF. 
  In order to rectify this, we show that the standard reduction from \Sat to \Sat that decreases the number of literals per clause to at most three maintains all the required properties. 
  The same is then shown for the standard reduction that transforms \CNF-formulas with at most
  three literals per clause into \EThreeCNF-formulas
  that have exactly three distinct literals per clause.
  Combining these three reductions, we can therefore satisfy all requirements of our theorem. 
\end{proof}

\begin{corollary}\label{cor:SATisNPhardInOurModel}
  Computing a satisfying assignment for a \EThreeCNF-formula whose
  one-clause-deleted
  subformulas all have a satisfying assignment from these assignments is \NP-hard.
\end{corollary}

\begin{proof}
  Given a \EThreeCNF-formula $\formula$, compute $f(\formula)$, where $f$ is the reduction
  from \cref{thm:3sat}. 
  Now compute $\sat(f(\formula) - {\clause})$ for every clause ${\clause}$ in $f(\formula)$
  and compute a satisfying assignment for $f(\formula)$ from these
  solutions. Use this assignment to determine
  whether $\formula$ is satisfiable.
\end{proof}

\section{Colorability}\label{sec:COL}

As mentioned in the previous section,
the constructions of some \DP-completeness results for minimality problems 
can be adapted to obtain \NP-hardness for computing optimal solutions
from optimal solutions to locally modified instances.
There are remarkably few complexity results for minimality problems; 
fortunately, however, 
\VertexMinimalThreeUnCol (the 
graphs that are not 3-colorable but that are 3-colorable after
deleting any vertex)\footnote{It
should be noted that \VertexMinimalThreeUnCol
is denoted by
$\textsc{Minimal}\textsc{-}\allowbreak{}\textsc{3}\textsc{-}\textsc{Un}\-\textsc{Color}\-\textsc{ability}$ by Cai and Meyer~\cite{cai-mey:j:dp} 
despite the fact that
minimality problems usually refer to the case of edge deletion.} is \DP-complete by reduction from \MinimalThreeUnSat~\cite{cai-mey:j:dp}.
We will show how to extract from said reduction a proof of the fact that 
computing an optimal coloring for a graph from optimal colorings 
for its one-vertex-deleted subgraphs is \NP-hard (\cref{thm:color-v}).
However, the standard notion of criticality is $\chi$-criticality under
edge deletion, 
and the construction by Cai and Meyer~\cite{cai-mey:j:dp} does
unfortunately not yield the analogous result for deleting edges. This was to be expected, since working with edge deletion is much harder. Surprisingly, however, a targeted modification of the constructed graph allows us to establish, through a far more elaborate case distinction, that 
computing an optimal coloring for a graph from optimal colorings 
for its one-edge-deleted subgraphs is \NP-hard (\cref{thm:color-e})
as well as that the related minimality problem 
\EdgeMinimalThreeUnCol is \DP-complete (\cref{thm:MU-DP-complete}).

\begin{lemma}\label{lem:color-v}
  There is a polynomial-time many-one reduction $g$ from \EThreeSat to \ThreeCol
  and a polynomial-time computable function $\opt$ such that, for every \EThreeCNF-formula
  $\formula$ and for every vertex $v$ in $g(\formula)$,
  $\opt(g(\formula) - v)$ is an optimal coloring for $g(\formula) - v$.
\end{lemma}

\begin{proof}
  Given a \EThreeCNF-formula $\formula$, let $g(\formula) = h(f(\formula))$, where $f$ is the reduction
  from \cref{thm:3sat} and $h$ is the reduction from \MinimalThreeUnSat
  to \VertexMinimalThreeUnCol by Cai and Meyer~\cite{cai-mey:j:dp}. 
  Since $h$ also reduces
  \EThreeSat to \ThreeCol~\cite[Lemma 2.2]{cai-mey:j:dp}, so does $g$. 
  A careful inspection of the reduction $g$  reveals that there is a polynomial-time computable function
  $\opt$ such that, for every vertex $v$ in $g(\formula)$, $\opt(g(\formula) - v)$ is a
  3-coloring of $g(\formula) - v$. We can also verify that  
  $g(\formula) - v$ does not have a 2-coloring, hence 
  $\opt(g(\formula) - v)$ is an optimal coloring. 
  We do not dive into the details as this lemma immediately follows
  from the proof of the analogous result for edge deletion
  (\cref{oldlem:color-e}), as explained in \cref{app:oldlem-color-e}. 
\end{proof}

\begin{theorem}\label{thm:color-v}
  Computing an optimal coloring for a graph from optimal colorings 
  for its one-vertex-deleted subgraphs is \NP-hard.
\end{theorem}

\begin{proof}
  Given a \EThreeCNF-formula $\formula$, compute $g(\formula)$, where $g$ is the reduction
  from \cref{lem:color-v}, compute $\opt(g(\formula) - v)$ for every vertex $v$ in $g(\formula)$,
  and from these optimal solutions compute one for $g(\formula)$. 
  This determines whether $g(\formula)$ is 3-colorable and 
  thus whether $\formula$ is satisfiable.
\end{proof}

\begin{lemma}\label{oldlem:color-e}
  There is a polynomial-time many-one reduction $g$ from \EThreeSat to \ThreeCol
  and a polynomial-time computable function $\opt$ such that, for every \EThreeCNF-formula
  $\formula$ and for every edge $e$ in $g(\formula)$,
  $\opt(g(\formula) - e)$ is an optimal coloring of $g(\formula) - e$.
\end{lemma}

\begin{proof}
  Given a \EThreeCNF-formula $\formula$, let $g(\formula) = h(f(\formula)) - e$,
  where $f$ is the reduction
  from \cref{thm:3sat}, $h$ is the reduction
  to \VertexMinimalThreeUnCol by Cai and Meyer~\cite{cai-mey:j:dp}, and $e$ is the edge  $\{v_\textnormal{c},v_\textnormal{s}\}$, with $v_\textnormal{c}$ being the unique vertex adjacent to all variable-setting vertices and $v_\textnormal{s}$ being the only remaining neighbor vertex of $v_\textnormal{c}$. 
  We prove in detail that $g$ has all the desired properties in \cref{app:oldlem-color-e}. See \cref{fig:caimeyertotal} for an example of the construction. 
\end{proof}

\begin{theorem}\label{thm:color-e}
  Computing an optimal coloring for a graph from optimal colorings 
  for its one-edge-deleted subgraphs is \NP-hard.
\end{theorem}

\begin{proof}
  The same argument as for \cref{thm:color-v} can be applied here.
\end{proof}

\begin{theorem}\label{thm:MU-DP-complete}
  \EdgeMinimalThreeUnCol is \DP-complete.
\end{theorem}

\begin{proof}
  Membership in \DP is immediate, since
  given a graph $G = (V,E)$, determining whether $G - e$
  is $3$-colorable for every $e\in E$ is in \NP
  and so is determining whether $G$ is $3$-colorable.
  As for \DP-hardness, the argument from the proof of \cref{oldlem:color-e}
  shows that mapping $\formula$ to $h(\formula) - \{v_\textnormal{c},v_\textnormal{s}\}$, where $h$
  is the reduction from \MinimalThreeUnSat to \VertexMinimalThreeUnCol
  by Cai and Meyer~\cite{cai-mey:j:dp}, gives a reduction from 
  \MinimalThreeUnSat to \EdgeMinimalThreeUnCol (and to \VertexMinimalThreeUnCol as well).
  Recall that \MinimalThreeUnSat is \DP-hard~\cite{pap-wol:j:facets}.
\end{proof}

Cai and Meyer~\cite{cai-mey:j:dp} show \DP-completeness for \VertexMinimalKayUnCol, 
for all $k \geq 3$. We now prove that the analogous result for deletion of edges holds as well.

\begin{theorem}\label{thm:EdgeMinimalKayUnCol}
  \EdgeMinimalKayUnCol is \DP-complete, for every $k \geq 3$.
\end{theorem}

\begin{proof}
  Membership in \DP is again immediate.  To show hardness for $k \geq 4$, we
  reduce \EdgeMinimalThreeUnCol to \EdgeMinimalKayUnCol.
  We use
  the construction for deleting vertices~\cite[Theorem 3.1]{cai-mey:j:dp} and map
  graph $G$ to $G+K_{k-3}$.\footnote{For two graphs $G_1$ and $G_2$, the \emph{graph join} $G_1+G_2$ is the disjoint union $G_1\cup G_2$ plus a \emph{join edge} added from every vertex of $G_1$ to every vertex of $G_2$; see, e.g., Harary's textbook on graph theory \cite[p. 21]{har:graph-theory}.} 
  Note that~$\chi(K_{k-3}) = k - 3$ and $\chi(H+H') = \chi(H) + \chi(H')$ for any two graphs $H$ and $H'$. 
  First suppose $G + K_{k-3}$ is in \EdgeMinimalKayUnCol.
  Then $G + K_{k-3}$ is not $k$-colorable, and so $G$ is not 3-colorable. 
  Let $e$ be an edge in $G$.
  Then $(G - e) + K_{k-3} = (G+K_{k-3}) - e$ is $k$-colorable, 
  and thus $G - e$ is 3-colorable.
  It follows that $G$ is in \EdgeMinimalThreeUnCol.
  
  Now suppose $G$ is in \EdgeMinimalThreeUnCol.
  Then $G + K_{k-3}$ is not $k$-colorable.
  Let $e$ be an edge in
  $G + K_{k-3}$. If $e$ is an edge in $G$, then
  $G - e$ is 3-colorable and so
  $(G + K_{k-3}) - e = (G - e) + K_{k-3}$ is $k$-colorable.
  If $e$ is an edge in $K_{k-3}$, then
  $K_{k-3} - e$ is $(k-4)$-colorable and
  $G$ is 4-colorable (let $\hat{e}$ be any edge in $G$, take a 3-coloring of $G - \hat{e}$,
  and change the color of one of the vertices incident to $\hat{e}$ to the remaining 
  color), so
  $(G + K_{k-3}) - e = G + (K_{k-3} - e)$ is $k$-colorable.
  Finally, if $e=\{v,w\}$ for a vertex $v$ in $G$ and a vertex $w$ in
  $K_{k-3}$, let $\hat{e}$ be an edge in $G$ incident to $v$, take a
  3-coloring of $G - \hat{e}$, take a disjoint $(k - 3)$-coloring
  of $K_{k-3}$, and change the color of $v$ to the color of $w$.
  As a result, for all edges $e$ in $G + K_{k-3}$, 
  $(G + K_{k-3}) - e$ is $k$-colorable.
  It follows that $G + K_{k-3}$ is in \EdgeMinimalKayUnCol.
\end{proof}

The construction above does not prove the 
analogues of \cref{lem:color-v,oldlem:color-e}:
Note that $G$ is $3$-colorable if and only if 
$(G + K_{k-3}) - v$ and $(G + K_{k-3}) - e$ are both  $(k-1)$-colorable for every vertex $v$ in $K_{k-3}$
and for every edge $e$ in $K_{k-3}$, and so we can certainly determine
whether a graph is 3-colorable from the optimal solutions to the one-vertex-deleted subgraphs and one-edge-deleted subgraphs of $G+K_{k-3}$ in polynomial time. 
Turning to criticality and vertex-criticality, 
we can bound their complexity as follows. 

\begin{theorem}\label{thm:criticalityandstability}
  The two problems of determining whether a graph is critical and whether it is vertex-critical are both in $\Theta_2^\textnormal{p}$ and \DP-hard. 
\end{theorem}

\begin{proof}
  For the $\Theta_2^\textnormal{p}$-membership of the two problems, we observe that the relevant chromatic numbers of a graph $G = (V,E)$ and its neighbors can be computed by querying the \NP oracle $\DecCol=\{(G,k)\,\mid\,\chi(G) \leq k\}$ 
  for every $(G,k)$, $(G - e, k)$, and $(G - v, k)$ for every $e\in E$, $v\in V$,
  and $k \leq \|V(G)\|$ in parallel.

  For the \DP-hardness of the two problems, we prove that $h(\formula)-\{v_\textnormal{c},v_\textnormal{s}\}$ is a reduction from \MinimalThreeUnSat to both of them. 
  We have already seen that it reduces \MinimalThreeUnSat to \EdgeMinimalThreeUnCol. 
  Hence, for every $\formula\in\MinimalThreeUnSat$, the graph $h(\formula)-\{v_\textnormal{c},v_\textnormal{s}\}$ is in  $\EdgeMinimalThreeUnCol\subseteq\VertexMinimalThreeUnCol$ and thus both critical and vertex-critical.
  For the converse it suffices to note that, for every $\formula\in\CNF$ with clauses of size at most 3,  $h(\formula)-\{v_\textnormal{c},v_\textnormal{s}\}$ is 4-colorable and thus in \EdgeMinimalThreeUnCol (in \VertexMinimalThreeUnCol, respectively) if and only if it is critical (vertex-critical, respectively).
\end{proof}

The exact complexity of these problems remains open, however. 
In particular, it is unknown whether they are 
$\Theta_2^\textnormal{p}$-hard.
This contrasts with the case of Vertex Cover, 
for which we prove in \cref{thm:VertexCriticalThetaComplete} that 
recognizing $\beta$-vertex-criticality is indeed $\Theta_2^\textnormal{p}$-complete.

Before that, however, we return to our model and consider Colorability under the local modification of adding an edge. If we allow only one query, the problem stays NP-hard via a simple Turing reduction: Iteratively adding edges to the given instance eventually leads to a clique as a trivial instance, see \cref{thm:onequerylemma} in \cref{app:onequerylemma}.  
Note that the restriction to one query is crucial for this reduction to work; without it, the branching may lead to an exponential blowup in the number of instances that need to be considered. The following theorem shows that this breakdown of the hardness proof is 
inevitable 
unless $\P=\NP$ 
since the problem becomes in fact polynomial-time solvable if just one more oracle call is granted. 

\begin{theorem}\label{thm:colorer}
  There is a polynomial-time algorithm
  that computes an optimal coloring for a graph from optimal colorings of all
  its one-edge-added supergraphs; in fact, two optimal colorings, one for each of two specific one-edge-added supergraphs, suffice.
\end{theorem}

For the proof of this theorem, we naturally extend the notion of universal vertices as follows.
\begin{definition}
  An edge $\{u,v\}\in E$ of a graph $G=(V,E)$ is called \emph{universal} if, for every vertex $x\in V\setminus\{u,v\}$, we have $\{x,u\}\in E$ or $\{x,v\}\in E$.
  A graph is called \emph{universal-edged} if all its edges are universal.
\end{definition}

Additionally, we denote, for any given graph $G = (V,E)$ and any vertex $x\in V$, the \emph{open neighborhood} of $x$ in $G$ by $N(x)\coloneqq\{y\mid\{x,y\}\in E\}$ and the \emph{closed neighborhood} of $x$ in $G$ by $N[x]\coloneqq N(x)\cup\{x\}$. 
We are now ready to give the proof of \cref{thm:colorer}. 

\begin{proof}[Proof of \cref{thm:colorer}]
  We show that {}\textsc{Colorer} (\cref{alg:maincolorer}), which uses 
  the oracle of our model 
  and {}\textsc{Subcol} (\cref{alg:subcolorer}) as subroutines, has the desired properties.
  
  \begin{algorithm}[htbp]
    \caption{{}\textsc{Colorer}}
    \label{alg:maincolorer}
    \textbf{Input:} An undirected graph $G=(V,E)$.\\
    \textbf{Output:} An optimal coloring for $G$.\\
    \textbf{Description:} Optimizes universal-edged graphs with two queries to {}\textsc{Oracle}, which provides optimal solutions to one-edge-added supergraphs; other graphs are optimized via {}\textsc{Subcol}.
    \begin{algorithmic}[1]
      \For{every edge $\{u,v\}\in E$}
      \For{every vertex $x\in V\setminus\{u,v\}$}
      \If{$\{u,x\}\notin E\,\wedge\,\{v,x\}\notin E$}
      \State $f_1\gets{}\textsc{Oracle}(G\cup\{u,x\})$
      \State $f_2\gets{}\textsc{Oracle}(G\cup\{v,x\})$
      \If{$f_1$ uses fewer colors on $G$ than $f_2$}
      \State \textbf{return} $f_1$ \label{line:returnFirst}
      \Else
      \State \textbf{return} $f_2$ \label{line:returnSecond}
      \EndIf
      \EndIf
      \EndFor
      \EndFor
      \State $k\gets 1$
      \While{${}\textsc{Subcol}(G,k)=\text{NO}$}
      \State $k\gets k+1$
      \EndWhile\vspace*{-.5ex}
      \State \textbf{return} ${}\textsc{Subcol}(G,k)$  \label{line:returnThird}
    \end{algorithmic}
  \end{algorithm}
  \vspace*{-3ex}\begin{algorithm}[htbp]
    \caption{{}\textsc{Subcol}}
    \label{alg:subcolorer}
    \textbf{Input:} An undirected, universal-edged graph $G=(V,E)$ and a positive integer $k$.\\
    \textbf{Output:} A $k$-coloring $f$ for $G$ if there is one; NO if there is none.\\
    \textbf{Description:} Works by recursion on $k$, with $k=1$ and $k=2$ serving as the base cases.    \begin{algorithmic}[1]
      \If{$G$ has no edge}
      \State \textbf{return} the constant 1-coloring with $f(x)=1$ for all $x\in V$.  \label{line:returnOne}
      \ElsIf{$k=1$}
      \State \textbf{return} NO. \label{line:returnTwo}
      \EndIf
      \If{$G$ has bipartition $\{A,B\}$}\vspace*{-2ex}
      \State \textbf{return} the 2-coloring $f(x)=\begin{cases}1&\text{ for }x\in A\text{ and}\\2&\text{ for }x\in B.\end{cases}$\label{line:returnThree}\vspace*{-2ex}
      \ElsIf{k=2}
      \State \textbf{return} NO. \label{line:returnFour}
      \EndIf
      \State Choose an arbitrary edge $\{\ell,r\}\in E$.
      \newlength{\maxwidth}
      \settowidth{\maxwidth}{$L$}
      \State $\makebox[\maxwidth][r]{$L$}\gets N(\ell)\setminus N[r]$;\quad $\makebox[\maxwidth][r]{$R$}\gets N(r)\setminus N[\ell]$;\quad $\makebox[\maxwidth][r]{$M$}\gets N(\ell)\cap N(r)$
      \State $\makebox[\maxwidth][r]{$g$}\gets \mathop{{}\textsc{Subcol}}(G[M],k-2)$
      \If{$g=\text{NO}$}
      \State \textbf{return} NO \label{line:returnFive}\vspace*{-3.5ex}
      \EndIf
      \State \textbf{return} the $k$-coloring $f(x)=
      \begin{cases}
      g(x)&\text{for $x\in M$,}\\
      k-1&\text{for $x\in L\cup\{r\}$, and}\\
      k&\text{for $x\in R\cup\{\ell\}$.}
      \end{cases}$ \label{line:returnSix}
    \end{algorithmic}
  \end{algorithm}
  
  \bigskip
  We begin by proving that {}\textsc{Colorer} is correct.
  Assume first that the input graph $G=(V,E)$ is not universal-edged. 
  Then {}\textsc{Colorer} can find an edge $\{u,v\}\in E$ with a non-neighboring vertex $x\in V$ and query the oracle on $G\cup\{u,x\}$ and $G\cup\{v,x\}$ for optimal colorings $f_1$ and $f_2$. We argue that at least one of them is also optimal for $G$.
  Let $f$ be any optimal coloring of $G$. Since $u$ and $v$ are connected by an edge, we have $f(u)\neq f(v)$ and hence $f(x)\neq f(u)$ or $f(x)\neq f(v)$; see \cref{fig:easyfigure} in \cref{app:moreaboutcolorer}. Thus, $f$ is also an optimal coloring of $G\add\{x,u\}$ or $G\add\{x,v\}$, and so we have $\chi(G)=\chi(G\add\{x,u\})$ or $\chi(G)=\chi(G\add\{x,v\})$.
  Therefore, $f_1$ or $f_2$ is an optimal coloring for $G$ as well and returned on \cref{line:returnFirst} or \ref{line:returnSecond}, respectively.
  
  The while loop can be entered only if the graph $G$ is universal-edged. This allows us to compute an optimal solution to $G$ with no queries at all by using ${}\textsc{Subcol}$ (\cref{alg:subcolorer}). 
  We will show that ${}\textsc{Subcol}$ is a polynomial-time algorithm that computes, for any universal-edged graph $G$ and any positive integer $k$, a $k$-coloring of $G$ if there is one, and outputs NO otherwise. 
  The while loop of ${}\textsc{Colorer}$ thus searches the smallest integer $k$ such that $G$ has a $k$-coloring, that is, $k=\chi(G)$.
  Hence, an optimal coloring of $G$ is returned on \cref{line:returnThird}.
  Due to $k=\chi(G)\le\|V\|$, {}\textsc{Colorer} has polynomial time complexity.
  
  It remains to prove the correctness and polynomial time complexity of {}\textsc{Subcol}. 
  This can be done by bounding its recursion depth and verifying the correctness for each of the six return statements; 
  this is hardest for the last two. The proof relies on the properties of the 
  partition 
  $M\cup L\cup R\cup \{\ell\}\cup \{r\}$ 
  as illustrated in \cref{fig:constructionexample}; 
  see \cref{app:moreaboutcolorer} for all details. 
\end{proof}

In this section, we have proven that \EdgeMinimalKayUnCol is complete for \DP for every $k \geq 3$ and demonstrated that Colorability remains \NP-hard in our model for deletion of vertices or edges, whereas it becomes polynomial-time solvable when the local modification is considered to be the addition of an edge.
In the next section, we turn our attention to Vertex Cover. 

\section{Vertex Cover}\label{sec:VC}

This section will show that the behavior of Vertex Cover in our model 
is distinctly different from the one that we demonstrated for Colorability
in the previous section.
In particular, \cref{thm:removevPoly} proves that computing an optimal
vertex cover from optimal solutions of one-vertex-deleted subgraphs can be done  
in polynomial time, which is impossible for optimal colorings according to \cref{thm:color-v} 
unless $\P=\NP$.

First, we note that the \NP-hardness proof for our most restricted case 
with only one query still works (i.e., \cref{thm:onequerylemma} in \cref{app:onequerylemma} is
applicable): Deleting vertices, adding edges, or
deleting edges repeatedly will always lead to the null graph, 
an edgeless graph, or a clique through polynomially many instances.
As we have seen for Colorability in the previous section, hardness proofs of this type may fail due to exponential branching as soon as multiple queries are allowed. We can show that, unless $\P=\NP$, this is necessarily the case for edge addition and vertex deletion since 
two granted queries suffice to obtain a polynomial-time algorithm.

\begin{theorem}\label{thm:removevPoly}
  There is a polynomial-time algorithm 
  that computes an optimal vertex cover for a graph from two optimal vertex covers for some one-vertex-deleted subgraphs.
\end{theorem}

\begin{proof}
  Observe what can happen when a vertex $v$ is removed from a graph $G$ with an optimal vertex cover of size $k$. If $v$ is part of any optimal vertex cover of $G$, then the size of
  an optimal vertex cover for $G-v$ is $k-1$.
  Given any graph $G$, pick any two adjacent vertices $v_1$ and $v_2$. Since there is an edge between them, one of them is always part of an optimal vertex cover, 
  thus either $G-v_1$ or $G-v_2$ or both will have an optimal vertex cover of size $k-1$. Two queries to the oracle  return optimal vertex covers for $G-v_1$ and $G-v_2$. The algorithm  chooses the smaller of these two covers (or any, if they are the same size) and adds the corresponding $v_i$. The resulting vertex cover has size $k$ and is thus optimal for $G$.
\end{proof}

\Cref{thm:addePoly} in \cref{app:addePoly} proves that the analogous result for adding an edge holds as well.

At this point, we would like to prove either an analogue to
\cref{thm:color-e}, showing that computing an optimal vertex cover 
is $\NP$-hard even if we get access to a solution for every one-edge-deleted
subgraph, or an analogue to \cref{thm:addePoly}, showing that the
problem is in $\P$ if we have access to a solution for more than one
one-edge-deleted subgraph.
We were unable to prove either, however. The latter is easy to do for many
restricted graph classes (e.g., graphs with bridges), 
yet we suspect that the problem is \NP-hard in general.
We will detail a few reasons for the apparent difficulty of proving this
statement after the following theorem and corollary, which look at deleting a
triangle as the local modification.

\begin{theorem}\label{thm:TriangleReduction}
  There is a reduction $g$ from \ThreeSat to \DecVC such that, for every \ThreeCNF-formula $\formula$ and
  for every triangle $T$ in $g(\formula)$, there is a polynomial-time computable
  optimal vertex cover of $g(\formula) - T$.
\end{theorem}

The proof of \cref{thm:TriangleReduction} relies on the standard reduction from 
\EThreeSat to \DecVC; see~\cite{gar-joh:b:int}, where clauses correspond to triangles; see \cref{app:TriangleReduction} for the details. 
Applying the same argument as in the proof of \cref{thm:color-v} yields the following corollary.

\begin{corollary}\label{cor:removeTrianglePoly}
  Computing an optimal vertex cover for a graph from optimal vertex covers 
  for the one-triangle-deleted subgraphs is \NP-hard.
\end{corollary}

What can we say about optimal vertex covers for one-edge-deleted graphs? 
Papadimitriou and Wolfe show~\cite[Theorem 4]{pap-wol:j:facets} that there is a reduction $g$ from 
\MinimalThreeUnSat to \MinimalNonVC (called $\textsc{Critical}\textsc{-}\textsc{Vertex}\-\textsc{Cover}$ in \cite{pap-wol:j:facets}; asking, 
given a graph $G$ and an
integer $k$, whether $G$ does not have a vertex cover of size
$k$ but all one-edge-deleted subgraphs do).
The reduction builds in a polynomial-time computable vertex
cover of size $k$ for every one-edge-deleted subgraph. And so 
$g$ is a reduction from \EThreeSat to \DecVC such that there
exists a polynomial-time computable function $\opt$ such that
for every \EThreeCNF-formula $\formula$
and $g(\formula) = (G,k)$, it holds, for every edge $e$ in $G$, that $\opt(G - e)$
is a vertex cover of size $k$.
Unfortunately, it may happen that 
an optimal vertex cover of $G - e$ has size $k-1$; 
namely, if $e$ is an edge connecting two triangles, 
an edge between two variable-setting vertices, or  
any edge of the clause triangles. 
The function $\opt$ does thus not give us an optimal vertex cover, 
thwarting the proof attempt.
This shows that we cannot always get results for our model 
from the constructions for criticality problems.

The following would be one approach to design a polynomial-time algorithm that computes an optimal vertex cover from optimal vertex covers for all one-edge-deleted subgraphs:
It is clear that deleting an edge does not increase the size of an optimal vertex cover and decreases it by at most one.
If, for any two neighbor graphs, the provided vertex covers differ in size, then we can take the smaller one, restore the deleted edge, and add any one of the two incident vertices to the vertex cover; this gives us the desired optimal vertex cover.
If the optimal vertex cover size decreases for all deletions of a single edge, we can do the same with any of them.
Thus, it is sufficient to design a polynomial-time algorithm that solves the problem on graphs whose one-edge-deleted subgraphs all have optimal vertex covers of the same size as an optimal vertex cover of the original graph.
One might suspect that only very few and simple graphs can be of this kind. However, we obtain infinitely many such graphs by the removal of any edge from different cliques, as already mentioned in the introduction.
In fact, there is a far larger class of graphs with this property and no apparent communality to be exploited for the efficient construction of an optimal vertex cover. 

We now turn to our complexity results of $\beta$-(vertex-)criticality. 
The reduction from 
\MinimalThreeUnSat to \MinimalNonVC by Papadimitriou and Wolfe~\cite{pap-wol:j:facets} 
establishes the \DP-hardness of deciding whether a graph
is $\beta$-critical. 
However, it seems unlikely that $\beta$-criticality is
in \DP. The obvious upper bound is $\Theta_2^\textnormal{p}$, since a polynomial
number of queries to a \DecVC oracle, namely $(G,k)$ and $(G-e,k)$ for all
edges $e$ in $G$ and all $k \leq \| V(G) \|$,
in parallel allows us to determine $\beta(G)$ and $\beta(G - e)$ for all 
edges $e$ in polynomial time, and thus allows us to determine
whether $G$ is $\beta$-critical. While we have not succeeded in proving
a matching lower bound, or even any lower bound beyond
\DP-hardness, we do get this lower bound for 
$\beta$-vertex-criticality, thereby obtaining the first
$\Theta_2^\textnormal{p}$-completeness result for a criticality problem.

\begin{theorem}\label{thm:VertexCriticalThetaComplete}
  Determining whether a graph is $\beta$-vertex-critical
  is  $\Theta_2^\textnormal{p}$-complete.
\end{theorem}

\begin{proof}
Membership follows with the same argument as above, this time
querying the oracle \DecVC in parallel for all 
$(G,k)$ and $(G-v,k)$ for all vertices $v$ in $G$ and all
$k \leq \| V(G) \|$.
To show that this problem is $\Theta_2^\textnormal{p}$-hard, we use a similar
reduction as the one by Hemaspaandra et al.~\cite[Lemma 4.12]{hem-spa-vog:j:kemeny}
to 
prove that it is $\Theta_2^\textnormal{p}$-hard to determine whether a given vertex is
a member of a minimum vertex cover.
We reduce from
the $\Theta_2^\textnormal{p}$-complete problem 
$\textsc{VC}_= = \{(G,H)\,\mid\,\beta(G) = \beta(H)\}$~\cite{wag:j:more-on-bh}.
Let $n = \max(\|V(G)\|,\|V(H)\|)$,
let $G'$ consist of $n+1-\|V(G)\|$ isolated vertices, 
let $H'$ consist of $n+1-\|V(H)\|$ isolated vertices, and
let $F = (G \cup G') + (H \cup H')$.
Note that $\beta(F) =
(n+1) + \min(\beta(G),\beta(H))$.
If $\beta(G) = \beta(H)$,
then $\beta(F) = (n+1) + \beta(G) = (n+1) + \beta(H)$ and for every vertex $v$
in $F$, $\beta(F - v) = n + \beta(G)$.
Thus, $F$ is critical.
If $\beta(G) \neq \beta(H)$, assume without loss of generality that
$\beta(G) <  \beta(H)$. Then $\beta(F) = n + 1 + \beta(G)$.
Let $v$ be a vertex in $G'$.
Then $\beta(F - v) = \min(n + 1 + \beta(G),
n + \beta(H)) = n + 1 + \beta(G)$, and therefore $F$ is not critical.
\end{proof}

\section{Conclusion and Future Research} \label{sec:conclusion}

We defined a natural model 
that provides new insights into the structural properties of \NP-hard problems. 
Specifically, we revealed interesting differences 
in the behavior of Colorability and Vertex Cover 
under different types of local modifications. 
While Colorability remains \NP-hard when the local modification is the deletion of either a vertex or an edge, 
there is an algorithm that finds an optimal coloring by querying the oracle on at most two edge-added supergraphs.
Vertex Cover, in contrast, becomes easy in our model for both deleting vertices and adding edges, as soon as two queries are granted. 
The question of what happens for the local modification of deleting an edge remains as an intriguing open problem that defies any simple approach, as briefly outlined above. 
Moreover, examples of problems where one can prove a jump from membership in \P to \NP-hardness at a given number of queries greater than $2$ might be especially instructive.

With its close connections to many distinct research areas, 
most notably the study of self-reducibility and critical graphs, 
our model can serve as a tool for new discoveries.
In particular, we were able to exploit the tight relations to 
criticality in the proof that recognizing 
$\beta$-vertex-critical graphs 
is $\Theta_2^\textnormal{p}$-hard, yielding the first 
completeness result for $\Theta_2^\textnormal{p}$ in the field.

\section*{Acknowledgments} \label{sec:acknowlegdements}

We thank the anonymous referees and Hans-Joachim Böckenhauer, Rodrigo R.\ Gumucio Escobar, Lane Hemaspaandra, Juraj Hromkovi\v{c}, Rastislav Kr\/{a}lovi\v{c}, Richard Kr\/{a}lovi\v{c}, 
Xavier Muñoz, Martin Raszyk, Peter Rossmanith, Walter Unger, and Koichi Wada for helpful comments and discussions.

\newpage

\begin{appendix}

\section{Tractability for Vertex Addition}\label{app:addvertextheorem}

We show that Vertex Cover and Colorability are trivially tractable if we are allowed to query an oracle for an optimal solution to the input graph with one vertex added. 

\begin{theorem}\label{thm:addvertextheorem}
  Under our model with the local modification of adding a vertex, the two problems of finding an optimal vertex cover and finding an optimal coloring are in \P.
\end{theorem}

\begin{proof}
  Let $G$ be the given graph. Add an isolated vertex $v$ and query the oracle for an optimal vertex cover or an optimal coloring of $G+v$, respectively. 
  Clearly, the restriction of this solution to $G$ is optimal as well since an isolated vertex is never part of an optimal vertex cover and can be colored arbitrarily. (Note that we could also avoid isolated vertices by adding a universal vertex instead since a universal vertex needs to be colored  
  in a unique color and is, without loss of generality, part of any optimal vertex cover.)
\end{proof}

\section{Hardness Results for Restriction to a Single Query}\label{app:onequerylemma}

In the most restricted case of our model, where we grant the algorithm but one single query for an optimal solution of a neighboring (i.e., locally modified) instance, all considered problems preserve their \NP-hardness.
The proof for this is simple enough and based on a technique that is commonly applied in reoptimization, self-reducibility, and many other fields. We formulate the following theorem as a generalization of Lemma 1 by~Böckenhauer et al.~\cite{boe-hro-etal:hardness-of-reoptimization}.

\begin{theorem}\label{thm:onequerylemma}
  Let $\OptProb$ be an optimization problem. 
  Let $T$ be a set of efficiently solvable instances of $\OptProb$.
  (Or, to be more precise: Let $T\subseteq\Sigma_1$ be a subset of instances 
  for the problem $\OptProb\subseteq\Sigma_1\times\Sigma_2$ such 
  that there is a polynomial-time algorithm that computes an optimal 
  solution on every instance of $T$.)
  Let the considered local modification be such that applying arbitrary
  local modifications repeatedly will inevitably transform any instance of
  $\OptProb$ into an instance in $T$ in a polynomial number of steps. 
  Then $\OptProb$ is \NP-hard in our model with restriction to one query.
\end{theorem}

\begin{proof}
  We give a reduction from $\OptProb$ in the classical setting to $\OptProb$ in our model.
  Let $\alg$ be a polynomial-time algorithm that computes on instance $I_i$
  a locally modified instance $I_{i+1}$ and uses an optimal solution to $I_{i+1}$ 
  to compute an optimal solution for $I_i$.
  For any given instance $I_0$ of $\OptProb$, 
  we thus get a chain $I_0,I_1,\ldots,I_n$ of polynomial length 
  such that $I_n$ is in $T$. 
  We can efficiently compute an optimal solution to $I_n\in T$ and then 
  use $\alg$ to successively compute optimal solutions to
  $I_{n-1},\ldots,I_1,I_0$ in polynomial time. 
\end{proof}

\section{Details of the Proof of \texorpdfstring{\cref{thm:3sat}}{Theorem~\ref{thm:3sat}}}\label{sec:satselfreduction}

This section goes through the reduction from \ThreeUnSat to \MinimalUnSat
by Papdimitriou and Wolfe~\cite{pap-wol:j:facets} (mostly following the
notation of Büning and Kullmann~\cite{kle-kul:b:handbook-sat-minunsat}) plus three standard reductions for additional form constraints, showing that the composition has all the properties that we need for our results.

\subsection{The main reduction: \ThreeUnSat to \MinimalUnSat}

Let $\Phi\in\EThreeCNF$ be a Boolean formula over the variable set $\{x_1,\ldots,x_n\}$, $n>1$; that is, $\Phi=C_1\wedge\ldots\wedge C_m$ with $C_i=\ell_{i,1}\vee \ell_{i,2}\vee \ell_{i,3}$ and $\ell_{i,j}\in\{x_1,\ldots,x_n,\overline{x}_1,\ldots,\overline{x}_n\}$, where an overline denotes negation. 
In the following, we construct in polynomial time an equivalent \CNF-formula $\Psi$ with the additional property that each one-clause-deleted subformula
of $\Psi$ has an easy-to-compute satisfying assignment.
First, delete without replacement any clause that contains a variable and its negation. (Such a clause is satisfied for every assignment.) Assume thus without loss of generality that no clause $C_i$ contains a variable and its negation.
Now introduce $m$ new variables $\{y_1,\ldots,y_m\}$ and let $\pi_i\coloneqq y_1\vee\ldots\vee y_{i-1}\vee y_{i+1}\vee\ldots\vee y_m$.
Let 
\[\Psi=\bigwedge_{i=1}^{m}(C_i\vee\pi_i)\wedge\bigwedge_{i=1}^{m}\bigwedge_{j=1}^3(\overline{\ell}_{i,j}\vee\pi_i\vee\overline{y}_i)\wedge\!\!\!\bigwedge_{1\le i<j\le m}\!\!\!(\overline{y}_i\vee\overline{y}_j).\]
This concludes the description of the construction. We now prove all the required properties. 

\subsubsection{\texorpdfstring{$\Phi$}{Phi} satisfiable \texorpdfstring{$\Rightarrow$}{iff} \texorpdfstring{$\Psi$}{Psi} satisfiable}

Let $\alpha$ be a satisfying assignment for $\Phi$. Then
\[
\beta\colon\begin{cases}
x_i\mapsto\alpha(x_i),&\text{ for all }i,\text{ and}\\
y_i\mapsto 0,&\text{ for }y\in Y, 
\end{cases}
\]
is a satisfying assignment for $\Psi$.
Indeed, we have $\beta(C_i\vee\pi_i)=1$ since $\alpha(C_i)=1$; and we have $\beta(\overline{\ell}_{i,j}\vee\pi_i\vee\overline{y}_i)=1$ and $\beta(\overline{y}_i\vee\overline{y}_j)=1$ since  $\beta(\overline{y}_i)=1$.

\subsubsection{\texorpdfstring{$\Psi$}{Psi} satisfiable \texorpdfstring{$\Rightarrow$}{iff} \texorpdfstring{$\Phi$}{Phi} satisfiable}

Let $\beta$ be a satisfying assignment for $\Psi$. We prove that $\beta$ (or, to be precise, the restriction $\beta|_X$) also satisfies $\Phi$. The satisfied clauses of the form $\overline{y}_i\vee\overline{y}_j$ in $\Psi$
guarantee that $\beta(\overline{y}_{\ihat})=0$ for at most one $\ihat\in\{1,\ldots,m\}$. 

\begin{description}
  \item[Case 1:]
  Assume that $\beta(\overline{y}_i)=1$
  for all $i\in\{1,\ldots,m\}$.
  All clauses that contain a literal $\overline{y}_i$ are clearly satisfied.
  The only remaining clauses have the form $C_i\vee\pi_i$. Moreover, $\beta(y_i)=0$ for all $i\in\{1,\ldots,m\}$ implies $\beta(\pi_i)=0$ for all $i\in\{1,\ldots,m\}$.
  Thus, we have $\beta(C_i\vee\pi_i)=1$ if and only if $\beta(C_i)=1$. By assumption $\beta$ satisfies all clauses of $\Psi$, in particular those of the form $C_i\vee\pi_i$; therefore $\beta$ also satisfies $\Phi$.
  
  \item[Case 2:]
  Assume that $\beta(\overline{y}_{\ihat})=0$, that is, $\beta(y_{\ihat})=1$ for exactly one $\ihat\in\{1,\ldots,m\}$.
  Then $\beta(\pi_{\ihat})=0$ and $\beta(\pi_i)=1$ for the remaining $i\neq\ihat$.
  Thus, all clauses that contain $y_i$ or $\pi_i$ with $i\neq\ihat$ are trivially satisfied.
  The only four remaining clauses are
  \[(C_{\ihat}\vee\pi_{\ihat})\wedge\bigwedge_{j=1}^3(\overline{\ell}_{\ihat,j}\vee\pi_{\ihat}\vee\overline{y}_{\ihat}),\]
  which, due to $\beta(\pi_{\ihat})=\beta(\overline{y}_{\ihat})=0$, simplify to
  $C_{\ihat}\wedge\overline{\ell}_{\ihat,1}\wedge\overline{\ell}_{\ihat,2}\wedge\overline{\ell}_{\ihat,3}$.
  This is unsatisfiable since $C_{\ihat}=\ell_{\ihat,1}\wedge \ell_{\ihat,2}\wedge \ell_{\ihat,3}$. Thus, case 2 cannot occur.
\end{description}

\subsubsection{\texorpdfstring{$\Psi$}{Psi} is satisfiable after deletion of an arbitrary clause}

There are three cases, which we handle separately.\vspace{1ex}
\begin{description}  
  \item[Case 1: The deleted clause is $\overline{y}_\ihat\vee\overline{y}_\jhat$.]
  We show that in this case, the following assignment is satisfying:
  \[
  \beta\colon\begin{cases}
  y_\ihat\mapsto 1,&\\
  y_\jhat\mapsto 1,&\\
  y_i\mapsto 0,&\text{ for }\ihat\neq i\neq \jhat,\text{ and}\\
  x_i\mapsto\text{arbitrary}.
  \end{cases}
  \]
  
  We have $\beta(\pi_i)=1$ for all $i\in\{1,\ldots,m\}$ since any $\pi_i$ contains either $y_\ihat$ or $y_\jhat$.
  The remaining clauses $\overline{y}_i\wedge\overline{y}_j$ with $(i,j)\neq(\ihat,\jhat)$ are trivially satisfied.\vspace{1ex}
  
  \item[Case 2: The deleted clause is $C_\ihat\vee\pi_\ihat$.]
  In this case, the assignment
  \[
  \beta\colon\begin{cases}
  y_\ihat\mapsto 1,&\\
  y_i\mapsto 0,&\text{ for }i\neq\ihat,\\
  x_i\mapsto 1,&\text{ for }x_i\in\{\ell_{\ihat,1},\ell_{\ihat,2},\ell_{\ihat,3}\},\\
  x_i\mapsto 0,&\text{ for }x_i\in\{\overline{\ell}_{\ihat,1},\overline{\ell}_{\ihat,2},\overline{\ell}_{\ihat,3}\},\text{ and}\\
  x_i\mapsto\text{arbitrary},&\text{ otherwise,}
  \end{cases}
  \]
  is satisfying. All clauses of the form $\overline{y}_i\vee\overline{y}_j$ are satisfied since only $y_\ihat$ is assigned 1 and $i\neq j$. We also have $\beta(\pi_i)=1$ for all $i\neq \ihat$, so all clauses containing $\pi_i$ for $i\neq\ihat$ are satisfied. Since $C_\ihat\vee\pi_\ihat$ is deleted, the only three remaining clauses are $\overline{\ell}_{\ihat,j}\vee\pi_\ihat\vee\overline{y}_\ihat$ for $j\in\{1,2,3\}$.
  These are satisfied because $\beta(\overline{\ell}_{\ihat,j})=1$ for $j\in\{1,2,3\}$. (Such an assignment is valid since no clause $C_i$ contains a variable and its negation, as mentioned in the first paragraph; in particular $C_\ihat$, that is, $\{\ell_{\ihat,1},\ell_{\ihat,2},\ell_{\ihat,3}\}\cap\{\overline{\ell}_{\ihat,1},\overline{\ell}_{\ihat,2},\overline{\ell}_{\ihat,3}\}=\emptyset$.)\vspace{1ex}
  
  \item[Case 3: The deleted clause is $\overline{\ell}_{\ihat,\jhat}\vee\pi_\ihat\vee\overline{y}_\ihat$.]
  Also in this case, the assignment
  \[
  \beta\colon\begin{cases}
  y_\ihat\mapsto 1,&\\
  y_i\mapsto 0,&\text{ for }i\neq\ihat,\\
  x_i\mapsto 1,&\text{ for }x_i\in\{\ell_{\ihat,\jhat}\}\cup\{\overline{\ell}_{\ihat,j}\mid j\neq\jhat\},\\
  x_i\mapsto 0,&\text{ for }x_i\in\{\overline{\ell}_{\ihat,\jhat}\}\cup\{\ell_{\ihat,j}\mid j\neq\jhat\},\\
  x_i\mapsto\text{arbitrary},&\text{ otherwise,}
  \end{cases}
  \]
  is satisfying.
  The same argument as in Case~2 shows that the assignment to  $y_{\ihat}$ satisfies all clauses but  the three clauses $C_\ihat$ and $\overline{\ell}_{\ihat,j}\vee\pi_\ihat\vee\overline{y}_\ihat$ for $j\in\{1,2,3\}\setminus\{\jhat\}$. The clause $C_\ihat$ is satisfied because $\beta(\ell_{\ihat,\jhat})=1$; the other two are satisfied due to $\beta(\ell_{\ihat,j})=1$ for $j\neq \jhat$.
\end{description}

\subsection{Additional Form Constraints}

\subsubsection{\CNF to \TwoOrThreeCNF}\label{sec:mureduction}

\TwoOrThreeCNF is the set of all $\CNF$-formulas with exactly two or three literals in every clause.
\paragraph*{Construction.}
$\Psi$ can be replaced by an equivalent
\TwoOrThreeCNF-formula $\Psi'$
while retaining the property that 
each one-clause-deleted subformula
has an easy-to-compute satisfying assignment.
We can use the standard reduction which replaces a clause $C_i=\ell_{i,1}\vee\ldots\vee \ell_{i,\|C_i\|}$ by
\[\underbrace{(\ell_{i,1}\vee z_{i,1})}_{C_{i,1}}\wedge\underbrace{(\overline{z}_{i,1}\vee \ell_{i,2}\vee z_{i,2})}_{C_{i,2}}\wedge\ldots\wedge\underbrace{(\overline{z}_{|C_i|-2}\vee \ell_{i,\|C_i\|-1}\vee z_{i,\|C_i\|})}_{C_{i,\|C_i\|-1}}\wedge\underbrace{(z_{i,\|C_i\|}\vee \ell_{\|C_i\|})}_{C_{i,\|C_i\|}},\]
where $z_{i,1},\ldots,z_{i,\|C_i\|}$ are $\|C_i\|$ new variables.

\paragraph*{Equivalence.}
$\Psi$ and $\Psi'$ are equivalent because we can use the assignment of truth values to $z_{i,1},\ldots,z_{i,\|C_i\|}$ to satisfy all but an arbitrary one of the substituted clauses above.

\paragraph*{Easy-to-compute satisfying assignments for one-clause-deleted subformulas.}
Deleting a clause $C_{\ihat,\jhat}$ from $\Psi'$ corresponds to the deletion of the clause $C_{\ihat}$ from $\Psi$ because
the clauses $C_{\ihat,j}$ with $j\neq\jhat$ can always be satisfied by 
assigning 1 to the variables $z_{\ihat,1},\ldots,z_{\ihat,\jhat-1}$ and 0 to the variables $z_{\ihat,\jhat+1},\ldots,z_{\ihat,\|C_\ihat\|}$.
  
\subsubsection{\EThreeCNF to \ThreeOccTwoOrThreeCNF}

\ThreeOccTwoOrThreeCNF is the set of \EThreeCNF-formulas where each variable occurs at most once in each clause and at most three times in the entire formula.

\paragraph*{Construction.}
Let $\Phi$ be a \EThreeCNF-formula over $\{x_1,\ldots,x_n\}$. Assume that $x_1$ occurs in $\Phi$ a total of $a$ times in the affirmative and $b$ times negated.
We replace the $a$ affirmative occurrences by $x_{1,1},\ldots,x_{1,a}$ and the $b$ negated occurrences by $\overline{x}_{1,a+1},\ldots,\overline{x}_{1,a+b}$.
Moreover, we add the following new clauses:
\[(\overline{x}_{1,1}\vee x_{1,2})\wedge\ldots\wedge (\overline{x}_{1,a+b-1}\vee x_{1,a+b}).\]
Repeating this for $x_2,\ldots,x_n$ results in a formula $\Psi$.
Observe that the added clauses are equivalent to the implication chain
\[x_{1,1}\Rightarrow x_{1,2}\Rightarrow\ldots\Rightarrow x_{1,a+b-1}\Rightarrow x_{1,a+b}\text{ for all $i\in\{1,\ldots,m\}$}.\]
We repeat this construction for $x_2,\ldots,x_n$ and obtain our formula $\Psi$. Now, we show that $\Psi$ is equivalent to $\Phi$.

\paragraph*{Correctness.}
Given a satisfying assignment $\alpha$ for $\Phi$, the assignment 
$\beta\colon
x_{i,j}\mapsto \alpha(x_i)
$
trivially satisfies the constructed formula~$\Psi$.
For the converse, assume there is a satisfying assignment $\beta$ for $\Psi$. 
We prove that the modified assignment  
$\beta'(x_{i,j})=\beta(x_{i,a})$ for all $j\in\{1,\ldots,a+b\}$ 
also satisfies $\Psi$. Obviously, $\beta'$ satisfies the implication chains since there is no dependence on $j$. To see that the other clauses are satisfied as well, consider the two possible assignments for $x_{i,a}$.
\begin{description}
  \item[Case 1.] If $\beta(x_{i,a})=1$, then $\beta(x_{i,j})=1$ for all $j\ge a$ by the implication chain. These variables are also assigned 1 by $\beta'$, which has $\beta'(x_{i,j})=1$ for all $j$. Thus $\beta'$ can only differ from $\beta$ on the variables $x_{i,j}$ with $j<a$. These are the positively occurring variables and $\beta'$ assigns 1 to all of them. Therefore, the changes to the assignment keep all the satisfied clauses satisfied.\vspace{1ex}
  
  \item[Case 2.] If $\beta(x_{i,a})=0$, then $\beta(x_{i,j})=0$ for all $j\le a$ by the contrapositive of the implication chain. These variables are also assigned 0 by $\beta'$, which has $\beta'(x_{i,j})=0$ for all $j$. Thus $\beta'$ can only differ from $\beta$ on the variables $x_{i,j}$ with $j>a$. These are the negatively occurring variables and $\beta'$ assigns 0 to all of them. As before, we conlcude that none of the changes to the assignment  renders any satisfied clause unsatisfied.\vspace{1ex}
\end{description}
Now, we trivially obtain from $\beta'$ a satisfying assignment for $\Phi$. 
\paragraph*{Easy-to-compute satisfying assignments for one-clause-deleted subformulas.}
Assume that a clause $\overline{x}_{\ihat,\jhat}\vee x_{\ihat,\jhat+1}$ is deleted. (For all other clauses, the correspondence between $\Phi$ and $\Psi$ is immediate.)
Then, the $\ihat$th implication chain breaks in two and we are left with 
\[x_{1,1}\Rightarrow\ldots\Rightarrow x_{\ihat,\jhat}\quad\text{ and }\quad x_{\ihat,\jhat+1}\Rightarrow\ldots\Rightarrow x_{1,a+b}.\]

Consider the four (partial) assignments
\begin{align*}
\beta_0\colon&\hspace{1em}x_{\ihat,j}\mapsto 0\text{ for all }j,\\
\beta_0'\colon&\begin{cases}
x_{\ihat,j}\mapsto 0\text{ for }j\neq \jhat,\\
x_{\ihat,\jhat}\mapsto 1,\\
\end{cases}\\
\beta_1\colon&\hspace{1em}x_{\ihat,j}\mapsto 1\text{ for all }j,\\
\beta_1'\colon&\begin{cases}
x_{\ihat,j}\mapsto 1\text{ for }j\neq \jhat+1,\\
x_{\ihat,\jhat+1}\mapsto 0.\\
\end{cases}
\end{align*}
  
As already seen, the two assignments $\beta_0$ and $\beta_1$ correspond to the possible assignments for $x_i$ in $\Phi$. The option of $\beta_0'$ and $\beta_1'$, however, allows us to freely switch the assignment to one variable, either $x_{\ihat,\jhat}$ or $x_{\ihat,\jhat+1}$. This means that the clause where this variable occurs can always be satisfied; which is tantamount to deleting this clause. For the remaining clauses, we use the assignment from \cref{sec:mureduction}.

\subsubsection{\ThreeOccTwoOrThreeCNF to \EThreeCNF}

With the following construction, we gain the property that every clause contains exactly 3 literals (instead of either 2 or 3), but lose the property that every variable occurs at most three times and every literal at most twice.
\paragraph*{Construction.}
Let $\Phi$ be a given formula in \ThreeOccTwoOrThreeCNF. We construct $\Psi$ in the following way: Clauses with exactly three literals remain unchanged. A clause $C_i=(\ell_{i,1}\vee \ell_{i,2})$ with two literals is replaced by $\widetilde{C}_i=(\ell_{i,1}\vee \ell_{i,2}\vee y_i)\wedge(\ell_{i,1}\vee \ell_{i,2}\vee \overline{y}_i)$, with a new variable $y_i$. A clause $C_i=(\ell_{i,1})$ with only one literal is replaced by 
\[\widetilde{C}_i=(\ell_{i,1}\vee y_i\vee z_i)\wedge(\ell_{i,1}\vee \overline{y}_i\vee z_i)\wedge(\ell_{i,1}\vee y_i\vee \overline{z}_i)\wedge(\ell_{i,1}\vee \overline{y}_i\vee \overline{z}_i),\] with new variables $y_i$ and $z_i$. 
\paragraph*{Correctness.}
By assigning the right values to $y_i$ and $z_i$, respectively, we satisfy any of the two clauses (any three of the four clauses, respectively) of $\widetilde{C}_i$, leaving one that simplifies to the original $C_i$. 
\paragraph*{Easy-to-compute satisfying assignments for one-clause-deleted subformulas.}
If a clause of $\widetilde{C}_i$ is deleted, we can again use the assignment to $y_i$ and $z_i$ to satisfy the remaining ones; thus virtually deleting the whole of $\widetilde{C}_i$.

\section{Full Proof of \texorpdfstring{\cref{oldlem:color-e}}{Lemma~\ref{oldlem:color-e}}}\label{app:oldlem-color-e}

For convenience, we restate \cref{oldlem:color-e} before giving its proof.  
\setcounterref{lemmaduplicate}{oldlem:color-e}
\addtocounter{lemmaduplicate}{-1}
\begin{lemmaduplicate}
  There is a polynomial-time many-one reduction $g$ from \EThreeSat to \ThreeCol
  and a polynomial-time computable function $\opt$ such that, for every \EThreeCNF-formula
  $\formula$ and for every edge $e$ in $g(\formula)$,
  $\opt(g(\formula) - e)$ is an optimal coloring of $g(\formula) - e$.
\end{lemmaduplicate}

\begin{proof}
  Given a \EThreeCNF-formula $\formula$, let $g(\formula) = h(f(\formula)) - \{v_\textnormal{c},v_\textnormal{s}\}$,
  where $f$ is the reduction
  from \cref{thm:3sat} and $h$ is the reduction from \MinimalThreeUnSat
  to \VertexMinimalThreeUnCol by Cai and Meyer~\cite{cai-mey:j:dp} described below.
  We will show that $g$
  reduces \EThreeSat to \ThreeCol
  and that there is a polynomial-time computable function $\opt$
  such that, for every \EThreeCNF-formula
  $\formula$ and for every edge $e$ in $g(\formula)$,
  $\opt(g(\formula) - e)$ is an optimal coloring of $g(\formula) - e$.
  
  For completeness, we briefly describe the reduction $h$ from
  \MinimalThreeUnSat
  to \VertexMinimalThreeUnCol~\cite{cai-mey:j:dp} 
  (also excellently explained by Rothe and Riege~\cite{rot-rie:j:completeness}).
  Let $\formula$ be a \EThreeCNF-formula with variables $\{x_1,\ldots,x_n\}$ 
  and clauses $\{c_1,\ldots,c_m\}$. 
  The graph $h(\formula)$ is defined as follows; see \cref{fig:caimeyertotal}. 
  First, we create two vertices $v_\textnormal{c}$ and $v_\textnormal{s}$ connected by an edge.
  Then, for each variable $x_i$, we create two vertices $x_i$ and
  $\overline{x}_i$ and connect them to each other and each of them to $v_\textnormal{c}$.
  For each clause ${\clause}_k=\ell_{k1}\vee \ell_{k2}\vee\ell_{k3}$ of $\formula$, 
  we create nine new vertices, namely a triangle $t_{k1}, t_{k2}, t_{k3}$ 
  and a pair $a_{ki},b_{ki}$ for each literal $\ell_{ki}$, 
  where $t_{ki}$ is connected to $b_{ki}$, $a_{ki}$ is connected
  to $b_{ki}$, and both $a_{ki}$ and $b_{ki}$ are connected to $v_\textnormal{s}$;
  if and only if the literal $\ell\in\{x_j,\overline{x}_j\}$ appears as the $i$th literal in
  ${\clause}_k$, there is an edge from $\ell$ to $a_{ki}$.
  
  \begin{figure}
    \begin{center}
      \begin{tikzpicture}[x=0.8cm,y=0.8cm,vertex/.style={draw,circle, inner sep=.0pt, minimum size=0.53cm},font=\scriptsize]
        \node[vertex] (vc) at (6.5,2.8) {$v_\textnormal{c}$};
        \node[vertex] (vs) at (6.5,-4) {$v_\textnormal{s}$};
        \node[vertex] (x_1) at (1.2*3-1+1.2*1+2.75-5.5,1.3) {$x_1$};
        \draw (x_1) edge [in=185, out=28] (vc);
        \node[vertex] (nx_1) at (1.2*3-1+1.2*1+2.75-4.5,1.3) {$\overline{x}_1$};
        \draw (nx_1) edge [in=195, out=25] (vc);
        \draw (nx_1) -- (x_1);
        \node at (1.2*3-1+1.2*1+2.75-3.25,1.3) {$\cdots$};
        \node[vertex] (x_i) at (1.2*3-1+1.2*1+2.75-2.25,1.3) {$x_i$};
        \draw (x_i) edge [in=210, out=40] (vc);
        \node[vertex] (nx_i) at (1.2*3-1+1.2*1+2.75-1.25,1.3) {$\overline{x}_i$};
        \draw (nx_i) edge [in=220, out=60] (vc);
        \draw (nx_i) -- (x_i);
        \node at (1.2*3-1+1.2*1+2.75,1.3) {$\cdots$};
        \node[vertex] (x_j) at (1.2*3-1+1.2*1+2.75+1.25,1.3) {$x_j$};
        \draw (x_j) edge [in=-40, out=120] (vc);
        \node[vertex] (nx_j) at (1.2*3-1+1.2*1+2.75+2.25,1.3) {$\overline{x}_j$};
        \draw (nx_j) edge [in=-30, out=140] (vc);
        \draw (nx_j) -- (x_j);
        \node at (1.2*3-1+1.2*1+2.75+3.25,1.3) {$\cdots$};
        \node[vertex] (x_n) at (1.2*3-1+1.2*1+2.75+4.5,1.3) {$x_n$};
        \draw (x_n) edge [in=-15, out=155] (vc);
        \node[vertex] (nx_n) at (1.2*3-1+1.2*1+2.75+5.5,1.3) {$\overline{x}_n$};
        \draw (nx_n) edge [in=-5, out=152] (vc);
        \draw (nx_n) -- (x_n);            
        \node[vertex] (a_11) at (2.4*1-3,-2) {$a_{11}$};
        \draw (a_11) edge[in=175,out=-30] (vs);
        \node[vertex] (b_11) at (2.4*1-2,-2) {$b_{11}$};
        \draw (b_11) edge[in=166,out=-30]  (vs);
        \draw (b_11) -- (a_11);
        \node[vertex] (a_21) at (2.4*2-3,-2) {$a_{12}$};
        \draw (a_21) edge[in=157,out=-30] (vs);
        \node[vertex] (b_21) at (2.4*2-2,-2) {$b_{12}$};
        \draw (b_21) edge[in=148,out=-30]  (vs);
        \draw (b_21) -- (a_21);
        \node[vertex] (a_31) at (2.4*3-3,-2) {$a_{13}$};
        \draw (a_31) edge[in=139,out=-50] (vs);
        \node[vertex] (b_31) at (2.4*3-2,-2) {$b_{13}$};
        \draw (b_31) edge[in=130,out=-60]  (vs);
        \draw (b_31) -- (a_31);
        \node[vertex] (a_12) at (2.4*1 +5.5,-2) {$a_{m1}$};
        \draw (a_12) edge[in=50,out=180+60] (vs);
        \node[vertex] (b_12) at (2.4*1 +6.5,-2) {$b_{m1}$};
        \draw (b_12) edge[in=41,out=180+50] (vs);
        \draw (b_12) -- (a_12);
        \node[vertex] (a_22) at (2.4*2 +5.5,-2) {$a_{m2}$};
        \draw (a_22) edge[in=32,out=180+30] (vs);
        \node[vertex] (b_22) at (2.4*2 +6.5,-2) {$b_{m2}$};
        \draw (b_22) edge[in=23,out=180+30] (vs);
        \draw (b_22) -- (a_22);
        \node[vertex] (a_32) at (2.4*3 +5.5,-2) {$a_{m3}$};
        \draw (a_32) edge[in=14,out=180+30] (vs);
        \node[vertex] (b_32) at (2.4*3 +6.5,-2) {$b_{m3}$};
        \draw (b_32) edge[in=5,out=180+30] (vs);
        \draw (b_32) -- (a_32);
        \node[vertex] (t11) at (3,-0.25) {$t_{11}$};
        \draw (t11)--(b_11);
        \node[vertex] (t12) at (3.5,-1.116) {$t_{12}$};
        \draw (t12)--(b_21);
        \draw (t12)--(t11);
        \node[vertex] (t13) at (4,-0.25) {$t_{13}$};
        \draw (t13)--(b_31);
        \draw (t12)--(t13);
        \draw (t13)--(t11);
        \node[vertex] (t21) at (11,-0.25) {$t_{m1}$};
        \draw (t21)--(b_12);
        \node[vertex] (t22) at (11.5,-1.116) {$t_{m2}$};
        \draw (t22)--(b_22);
        \draw (t22)--(t21);
        \node[vertex] (t23) at (12,-0.25) {$t_{m3}$};
        \draw (t23)--(b_32);
        \draw (t22)--(t23);
        \draw (t23)--(t21);
        \node at (1.2*3-1+1.2*1+2.75,-1.2) {$\cdots$};
        \draw (x_1)--(a_11);
        \draw (x_i) edge [in=80, out=190](a_21);
        \draw (x_j)--(a_12);
        \draw (x_n)edge [in=90, out=-120](a_22);
        \draw (nx_j)edge[bend right=10](a_31);
        \draw (nx_n)edge [in=90, out=-60](a_32);
        \draw[dashed, rounded corners] (-1.2, -2.5) rectangle (5.8, .3) {};
        \draw[dashed, rounded corners] (7.3, -2.5) rectangle (14.3, .3) {};
        \node at (-.85,0) {${\clause}_1$};
        \node at (13.9,0) {${\clause}_m$};
      \end{tikzpicture}
    \end{center}
    \caption{The graph $h(\formula) - \{v_\textnormal{c},v_\textnormal{s}\}$ for a \EThreeCNF-formula with 
      ${\clause}_1 = x_1 \vee x_i \vee \overline{x}_j$ and 
      ${\clause}_m = x_j \vee x_n \vee \overline{x}_n$.}\label{fig:caimeyertotal}
  \end{figure}
  \medskip
  
  We first show that $g$ is a reduction from \EThreeSat to \ThreeCol.
  Cai and Meyer~\cite[Lemma 2.2]{cai-mey:j:dp} show that
  $h$ is a reduction from \EThreeSat to \ThreeCol.
  This implies that
  for every \EThreeCNF-formula $\formula$,
  $\formula$ is satisfiable if and only if $h(f(\formula))$ is 3-colorable, so it
  suffices to show that if $h(f(\formula)) - \{v_\textnormal{c},v_\textnormal{s}\}$ is 3-colorable, then
  so is $h(f(\formula))$. Consider a 3-coloring of $h(f(\formula)) - \{v_\textnormal{c},v_\textnormal{s}\}$
  such that $v_\textnormal{c}$ and $v_\textnormal{s}$ get the same color. Following the original proof~\cite{cai-mey:j:dp},
  we call the colors $\T$, $\F$, and $\C$. Assume that $v_\textnormal{c}$ and $v_\textnormal{s}$ are colored $\T$.
  Now change the color of $v_\textnormal{c}$ to $\C$ and change the color of every literal vertex originally
  colored $\C$ to $\T$. It is easy to check that this new coloring is a 3-coloring
  of $h(f(\formula))$.\footnote{Note that this also shows that deleting the
  edge $\{v_\textnormal{c},v_\textnormal{s}\}$ is crucial for the lemma and that
  the original construction~\cite{cai-mey:j:dp} does not work for deleting edges.}
  
  \medskip
  
  Let $e$ be an edge in $g(\formula)$. We need to show that there is a polynomial-time
  computable optimal coloring of $g(\formula) - e$. We show that there is a polynomial-time computable 3-coloring. (This is optimal because $g(\formula) - e$ is not 2-colorable since it contains triangles.)
  
  Let ${\clause}$ be a clause in $f(\formula)$.
  Let $\alpha$
  be a polynomial-time computable assignment for $f(\formula) - {\clause}$.
  From this assignment, we can compute in polynomial time a 3-coloring of 
  $g(\formula) - {\clause}$, i.e., $g(\formula)$ minus the nine clause-vertices representing ${\clause}$,
  in such a way that the literal-vertices are colored $\T$ or $\F$ according to $\alpha$,
  $v_\textnormal{c}$ is colored $\C$, and $v_\textnormal{s}$ is colored $\T$.
  
  \begin{enumerate}
  \item If $e = \{x_i,\overline{x}_i\}$, 
  let ${\clause}$ be a clause in $f(\formula)$ such that $x_i$ occurs
  positively in ${\clause}$ (note that it follows from
  the definition of $f$  that every literal appears positively in at
  least one clause of $f(\formula)$).
  Color $g(\formula) - {\clause}$ as explained above.
  If $x_i$ is colored $\F$, change its color to $\T$. This is still a
  3-coloring of $g(\formula) - {\clause}$, and since $x_i$ occurs positively in ${\clause}$, we can extend
  this coloring to a 3-coloring of $g(\formula)$.
  
  \item If $e = \{v_\textnormal{c},\ell_i\}$, where $\ell_i \in \{x_i, \overline{x}_i\}$, 
  let ${\clause}$ be a clause in $f(\formula)$ such that $\ell_i$ occurs
  positively in ${\clause}$.  Color $g(\formula) - {\clause}$ as explained above.
  If $\ell_i$ is colored $\T$, then we can extend the coloring to a 3-coloring
  of $g(\formula)$ in polynomial time.
  If $\ell_i$ is colored $\F$, change the color of $\ell_i$ to $\C$, and
  for every $a$-vertex connected to $\ell_i$, change its color from $\C$ to $\F$, and
  for every $b$-vertex connected to a changed $a$-vertex, change its color from
  $\F$ to $\C$. It is possible that because of this, the $b$-vertices in a clause are
  all colored $\C$, and the attached triangle cannot be colored.
  If that is the case, there is a $b$-vertex in the clause that is connected to 
  an $a$-vertex that is connected to a literal that is colored $\T$.
  Change the color of the $a$-vertex to $\C$ and that of the $b$-vertex to $\F$. Now
  we can color the triangle.  This results in a 3-coloring of
  $g(\formula) - {\clause}$, and it is easy to check that
  we can extend this coloring to a 3-coloring of $g(\formula)$.
  
  \item Let ${\clause}$ be a clause such that $e$ is connected to a clause vertex of ${\clause}$.
  Again, color $g(\formula) - {\clause}$ as above.
  If $\alpha$ satisfies $f(\formula)$, we can in polynomial time compute a 3-coloring of $g(\formula)$.
  So suppose that $\alpha$ does not satisfy $f(\formula)$.
  Then all literal-vertices connected to a clause-vertex of ${\clause}$ are colored $\F$.
  When we try to extend this coloring, all $a$-vertices in the clause must
  be colored $\C$ and all $b$-vertices must be colored $\F$, which means that we
  cannot color the triangle with 3 colors.
  If $e$ is one of the triangle edges, we can color the triangle-vertices
  $\T$, $\C$, and $\T$. If $e$ connects a $b$-vertex to a $t$-vertex, we can color that
  $t$-vertex $\F$, and the other $t$-vertices $\T$ and $\C$.
  For the remaining cases,
  we show that we can change the color of one of the $b$-vertices, which again allows 
  us to color the triangle.
  If $e$ connects a literal to an $a$-vertex, we can color the $a$-vertex $\F$
  and the connecting $b$-vertex $\C$.
  If $e$ connects an $a$-vertex to a $b$-vertex, we can color the $b$-vertex $\C$.
  If $e$ connects $v_\textnormal{s}$ to an $a$-vertex, we can color the $a$-vertex $\T$
  and the connecting $b$-vertex $\C$.
  If $e$ connects $v_\textnormal{s}$ to a $b$-vertex, we can color the $b$-vertex $\T$.
  \end{enumerate}
  
  This completes the proof of \cref{oldlem:color-e}. 
  We now
  explain why $g$ also fulfills the requirements of \cref{lem:color-v}.
  Let $v$ be a vertex in $g(\formula)$, and let $e$ be an edge incident
  with $v$ (such an edge always exists since $g(\formula)$ does not contain isolated vertices).
  Then $\opt(g(\formula) - e)$ is a 3-coloring. 
  This gives us a 3-coloring of $g(\formula)-v$, which is optimal 
  since $g(\formula) - v$ does not have a 2-coloring.
\end{proof}

\section{Additional Explanations for the Proof of \texorpdfstring{\cref{thm:colorer}}{Theorem~\ref{thm:colorer}}}\label{app:moreaboutcolorer}

This appendix provides supplementing material for the proof of \cref{thm:colorer}. 
On the one hand, we give \cref{fig:easyfigure}, which illustrates how \cref{alg:maincolorer} can optimally color graphs that are not universal-edged using only two queries.
On the other hand, we prove the correctness and polynomial-time complexity of \textsc{Subcol}, which is used to optimally color universal-edged graphs without any queries, in \cref{lem:subcolorer} and exemplify the used construction in \cref{fig:constructionexample}. 

\begin{figure}
  \begin{center}
    \begin{subfigure}[c]{0.32\textwidth}
      \begin{center}
        \begin{tikzpicture}[xscale=.8,yscale=.8]
          \tikzset{every node/.append style={minimum size=0.43cm, draw,circle,font=\normalfont,inner sep=0.05cm}}%
          \coordinate (u) at (-1,0);
          \coordinate (v) at (1,0);
          \coordinate (x) at (0,2);
          \draw[thick] (u) -- (v);
          \node[fill=gray!50,thin] at (u) {$u$};
          \node[fill=gray!0,thin] at (v) {$v$};
          \node[fill=gray!50,thin] at (x) {$x$};
          \draw[dashed] (0,.84) ellipse (2 and 2);
          \node[draw=none] at (-2,2.7) {};
        \end{tikzpicture}
        \captionsetup{width=.6\linewidth}
        \subcaption{An arbitrary $k$-color\-ing of~$G$.}
      \end{center}
    \end{subfigure}
    \begin{subfigure}[c]{0.32\textwidth}
      \begin{center}
        \begin{tikzpicture}[xscale=.8,yscale=.8]
          \tikzset{every node/.append style={minimum size=0.43cm, draw,circle,font=\normalfont,inner sep=0.05cm}}%
          \coordinate (u) at (-1,0);
          \coordinate (v) at (1,0);
          \coordinate (x) at (0,2);
          \draw[thick] (u) -- (v);
          \draw[thick] (u) -- (x) node[midway,left,draw=none,xshift=+.1em,yshift=.6ex] {$\{u,x\}$};
          \node[fill=gray!50,thin] at (u) {$u$};
          \node[fill=gray!0,thin] at (v) {$v$};
          \node[fill=gray!50,thin] at (x) {$x$};
          \draw[dashed] (0,.84) ellipse (2 and 2);
          \node[draw=none] at (-2,2.7) {};
        \end{tikzpicture}
        \captionsetup{width=.6\linewidth}
        \subcaption{For $G\add\{u,x\}$, this is no longer a $k$-coloring.}
      \end{center}
    \end{subfigure}
    \begin{subfigure}[c]{0.32\textwidth}
      \begin{center}
        \begin{tikzpicture}[xscale=.8,yscale=.8]
          \tikzset{every node/.append style={minimum size=0.43cm, draw,circle,font=\normalfont,inner sep=0.05cm}}%
          \coordinate (u) at (-1,0);
          \coordinate (v) at (1,0);
          \coordinate (x) at (0,2);
          \draw[thick] (u) -- (v);
          \draw[thick] (v) -- (x) node[midway,right,draw=none,xshift=-.1em,yshift=.6ex] {$\{v,x\}$};
          \node[fill=gray!50,thin] at (u) {$u$};
          \node[fill=gray!0,thin] at (v) {$v$};
          \node[fill=gray!50,thin] at (x) {$x$};
          \draw[dashed] (0,.84) ellipse (2 and 2);
          \node[draw=none] at (-2,2.7) {};
        \end{tikzpicture}
        \captionsetup{width=.6\linewidth}
        \subcaption{However, it remains a $k$-coloring for  $G\add\{v,x\}$.}
      \end{center}
    \end{subfigure}
  \end{center}
  \caption{Any $k$-coloring of $G$ is also a $k$-coloring for  $G\add\{u,x\}$ or $G\add\{v,x\}$ or both. If a coloring is optimal for $G$, then it is also optimal for $G\add\{u,x\}$ or $G\add\{v,x\}$.
    The figure depicts only the induced subgraphs of $\{u,v,x\}$.}\label{fig:easyfigure}
\end{figure}
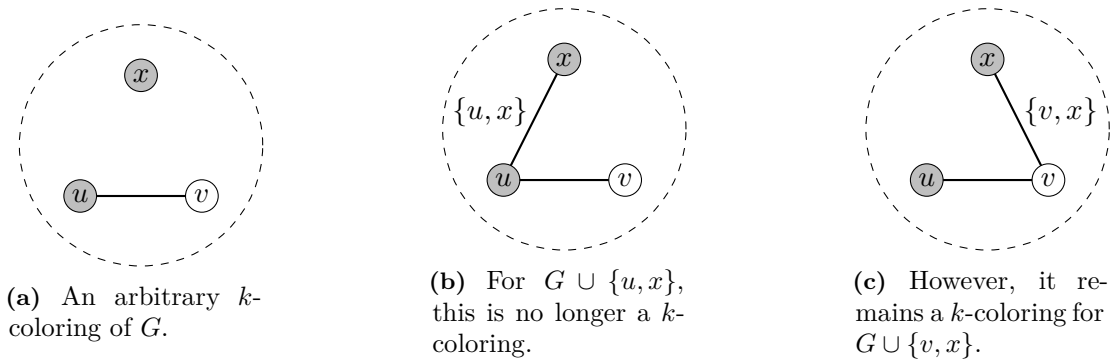

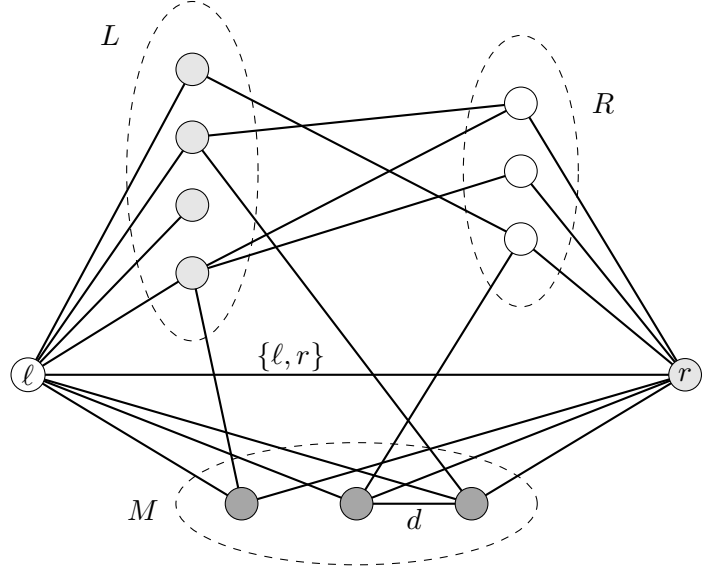
\begin{figure}
  \begin{center}
    \begin{tikzpicture}[xscale=1.2,yscale=1.2]
      \tikzset{xscale=.9,yscale=0.75,every node/.append style={minimum size=0.43cm, draw,circle,font=\normalfont,inner sep=0.05cm}}%
      \coordinate (l) at (-4,0);
      \coordinate (r) at (4,0);
      \coordinate (l1) at (-2,1.5);
      \coordinate (l2) at (-2,2.5);
      \coordinate (l3) at (-2,3.5);
      \coordinate (l4) at (-2,4.5);
      \coordinate (r1) at (2,2);
      \coordinate (r2) at (2,3);
      \coordinate (r3) at (2,4);
      \coordinate (m1) at (-1.4,-1.9);
      \coordinate (m2) at (0,-1.9);
      \coordinate (m3) at (1.4,-1.9);
      \draw[thick] (l) -- (r) node[midway,below,xshift=-7ex,yshift=6.5ex,draw=none] {$\phantom{e=}\{\ell,r\}$};
      \draw[thick] (l) -- (l1);
      \draw[thick] (l) -- (l2);
      \draw[thick] (l) -- (l3);
      \draw[thick] (l) -- (l4);
      \draw[thick] (r) -- (r1);
      \draw[thick] (r) -- (r2);
      \draw[thick] (r) -- (r3);
      \draw[thick] (l1) -- (r2);
      \draw[thick] (l3) -- (r3);
      \draw[thick] (l4) -- (r1);
      \draw[thick] (l1) -- (r3);
      \draw[thick] (l) -- (m1) -- (r);
      \draw[thick] (l) -- (m2) -- (r);
      \draw[thick] (l) -- (m3) -- (r);
      \draw[thick] (l3) -- (m3);
      \draw[thick] (r1) -- (m2);
      \draw[thick] (l1) -- (m1);
      \draw[thick] (m2) -- (m3) node[midway,below,draw=none,yshift=.3ex] {$d$};
      \node[fill=gray!0,thin] at (l) {$\ell$};
      \node[fill=gray!20,thin] at (r) {$r$};
      \node[fill=gray!20,thin] at (l1) {};
      \node[fill=gray!20,thin] at (l2) {};
      \node[fill=gray!20,thin] at (l3) {};
      \node[fill=gray!20,thin] at (l4) {};
      \node[fill=gray!0,thin] at (r1) {};
      \node[fill=gray!0,thin] at (r2) {};
      \node[fill=gray!0,thin] at (r3) {};
      \node[fill=gray!70,thin] at (m1) {};
      \node[fill=gray!70,thin] at (m2) {};
      \node[fill=gray!70,thin] at (m3) {};
      \draw[dashed] (-2,3) ellipse (.8 and 2.5);
      \node[draw=none] at (-3,5) {$L$};
      \draw[dashed] (2,3) ellipse (.7 and 2);
      \node[draw=none] at (3,4) {$R$};
      \draw[dashed] (0,-1.9) ellipse (2.2 and .9);
      \node[draw=none] at (-2.6,-2) {$M$};
    \end{tikzpicture}
  \end{center}
  \caption{An example of the construction that we use in \textsc{Subcol} (\cref{alg:subcolorer}) for a $k$-colorable graph $G$, exploiting the fact that $G$ is known to be universal-edged. In the example, we have $k=4$. In general, the graph $G$ is $k$-colorable if and only if the induced subgraph $G[M]$ is $(k-2)$-colorable. The subgraphs $G[L]$ and $G[R]$ are independent sets. 
    The following relations
    hold in general as well:
    \[\protect\begin{aligned}
    L={}&N(\ell)\setminus N[r]=V\setminus N[r],&
    M={}&N(\ell)\cap N(r),\\
    R={}&
    N(r)\setminus N[\ell]=V\setminus N[\ell],\text{ and }&V={}&N[\ell]\cup N[r]=L\cup M\cup R\cup\{\ell,r\}.
    \protect\end{aligned}.\]
    In the example, only edge $d$ prevents $G[M]$ from being $1$-colorable and thus $G$ from being $3$-colorable.
  }
  \label{fig:constructionexample}
\end{figure}

\begin{lemma}\label{lem:subcolorer}
  The subroutine \textsc{Subcol} (\cref{alg:subcolorer}) is correct and runs in polynomial time.
\end{lemma}

For convenience in referencing the lines of the algorithm \textsc{Subcol}, we reprint it here.

\setcounterref{algorithm}{alg:subcolorer}
\addtocounter{algorithm}{-1}
\begin{algorithm}[H]
  \caption{{}\textsc{Subcol}}
  \textbf{Input:} An undirected, universal-edged graph $G=(V,E)$ and a positive integer $k$.\\
  \textbf{Output:} A $k$-coloring $f$ for $G$ if there is one; NO if there is none.\\
  \textbf{Description:} Works by recursion over $k$, with $k=1$ and $k=2$ serving as the base cases.    \begin{algorithmic}[1]
    \If{$G$ has no edge}
    \State \textbf{return} the constant 1-coloring with $f(x)=1$ for all $x\in V$.
    \ElsIf{$k=1$}
    \State \textbf{return} NO.
    \EndIf
    \If{$G$ has bipartition $\{A,B\}$}\vspace*{-2ex}
    \State \textbf{return} the 2-coloring $f(x)=\begin{cases}1&\text{ for }x\in A,\text{ and}\\2&\text{ for }x\in B.\end{cases}$\vspace*{-2ex}
    \ElsIf{k=2}
    \State \textbf{return} NO.
    \EndIf
    \State Choose an arbitrary edge $\{\ell,r\}\in E$.
    \newlength{\maxwidthtwo}
    \settowidth{\maxwidthtwo}{$L$}
      \State $\makebox[\maxwidthtwo][r]{$L$}\gets N(\ell)\setminus N[r]$;\quad $\makebox[\maxwidthtwo][r]{$R$}\gets N(r)\setminus N[\ell]$;\quad $\makebox[\maxwidthtwo][r]{$M$}\gets N(\ell)\cap N(r)$
      \State $\makebox[\maxwidthtwo][r]{$g$}\gets \mathop{{}\textsc{Subcol}}(G[M],k-2)$
      \If{$g=\text{NO}$}
    \State \textbf{return} NO\vspace*{-3.5ex}
    \EndIf
    \State \textbf{return} the $k$-coloring $f(x)=
    \begin{cases}
    g(x)&\text{for $x\in M$,}\\
    k-1&\text{for $x\in L\cup\{r\}$, and}\\
    k&\text{for $x\in R\cup\{\ell\}$.}
    \end{cases}$
  \end{algorithmic}
\end{algorithm}

\begin{proof}[Proof of \cref{lem:subcolorer}]
  Note first that the input for {}\textsc{Subcol} is a pair $(G,k)$, where $G$ is a universal-edged graph and $k$ a positive integer.
  It is thus clear that {}\textsc{Subcol} runs in polynomial time: The recursion depth is $\lfloor(k-1)/2\rfloor$ and for the two base cases it is easy to check whether $G$ has edges and whether $G$ has a bipartition and, if there is any, find one in polynomial time.
  
  We now show that {}\textsc{Subcol} is correct by going through all six return statements.
  The first one in \cref{line:returnOne} is correct since the constant coloring is a $k$-coloring for any $k\in\N\setminus\{0\}$. With the second one in \cref{line:returnTwo}, the case $k=1$ is completely and correctly covered. Analogously, the third one in \cref{line:returnThree} is correct since a $2$-coloring is a $k$-coloring for any $k\in\N\setminus\{0,1\}$, and the case $k=2$ is correctly covered together with the fourth return statement in \cref{line:returnFour}.
  If none of the first four return statements of {}\textsc{Subcol} are executed, the graph $G$ has an edge and the choice of an edge $\{\ell,r\}\in E$ is possible.
  
  For the last two return statements in \cref{line:returnFive,line:returnSix}, we will prove the correctness by induction over $k$, with the above two cases $k=1$ and $k=2$ serving as the induction basis. We will rely on the properties of a partition of $G$ that we describe in what follows;
  see \cref{fig:constructionexample} for an illustrating example with a graph that is $k$-colorable for $k=4$ but not for $k=3$.
  Let $\{\ell,r\}$ be the edge of $G$ as chosen by the algorithm. The remaining vertices $V\setminus\{\ell,r\}$ are partitioned, depending on the way they are connected to $\{\ell,r\}$, into the three sets $L$, $R$, and $M$: $L$ contains the vertices adjacent to $l$ but not to $r$, $R$ contains the vertices connected to $r$ but not to $l$, and $M$ contains the vertices that are adjacent to both $l$ and $r$.
  Note that the sets $L$, $R$, and $M$ are disjoint. They cover $V\setminus\{\ell,r\}$ since every vertex is adjacent to $l$ or $r$ because $G$ is universal-edged and $\{\ell,r\}$ thus  universal.
  We now consider the case that NO is returned with the fifth return statement. This happens only if $g=\mathop{{}\textsc{Subcol}}(G[M],k-2)=\textrm{NO}$. Thus, $G[M]$ is not $(k-2)$-colorable by the induction hypothesis. We show that a $k$-coloring of $G$ yields a $(k-2$)-coloring of $G[M]$, thus proving by contradiction that $G$ is not $k$-colorable.
  Assume that there is a $k$-coloring of $G$.
  Due to the edge $\{\ell,r\}$, the two vertices $\ell$ and $r$ have two different colors out of the $k$ available ones.
  Since all vertices of $M$ are adjacent to both $\ell$ and $r$, the subgraph $G[M]$ is indeed colored by $f$ with the $k-2$ remaining colors.
  
  Finally, we consider the case where the last return statement in \cref{line:returnSix} is reached. 
  We need to prove that 
  the output $f(x)$  
  is a $k$-coloring of $G$. By the induction hypothesis, we know that $g$ is a $(k-2)$-coloring on $G[M]$ using the colors ${1,\ldots,k-2}$.
  The remaining vertices $L\cup\{r\}$ and $R\cup\{\ell\}$ are colored with $k-1$ and $k$, respectively. 
  Thus, it suffices to show that $G[L\cup\{r\}]$ and $G[R\cup\{\ell\}]$ are independent sets.
  Consider first $G[L\cup\{r\}]$. On the one hand, none of the vertices in $L$ are adjacent to $r$ by the definition of $L$. On the other hand, if there were $x,y\in L$ with $\{x,y\}\in E$, this would contradict the universality of $\{x,y\}$ for $r$. 
  Analogously, we see that $G[R\cup\{\ell\}]$ is an independent set, concluding the proof.
\end{proof}

\section{Analogue of \texorpdfstring{\cref{thm:removevPoly}}{Theorem~\ref{thm:removevPoly}}}\label{app:addePoly}

We prove the analogue of \cref{thm:removevPoly} 
for adding edges instead of deleting vertices.

\begin{theorem}\label{thm:addePoly}
  There is a polynomial-time algorithm that computes an optimal vertex cover for a graph from two optimal vertex covers for some one-edge-added supergraphs.
\end{theorem}

\begin{proof}
  Observe first what can happen when an edge $e$ is added to a graph $G$ with an optimal vertex cover of size $k$. If one of its endpoints $v$ is part of any optimal vertex cover of $G$, then the optimal vertex cover of $G$ containing $v$ is also
  an optimal vertex cover for $G\add e$.
  Given any graph $G$, the algorithm picks any two non-universal vertices $v_1$ and $v_2$ that are adjacent. Since there is an edge between them, any given optimal vertex cover contains at least one of them. If edges are added to this vertex, the vertex cover of size $k$ thus remains optimal. Because $v_1$ and $v_2$ are non-universal, the algorithm can add to $G$ an edge $e_1$ that is incident to $v_1$ and an edge $e_2$ incident to $v_2$. Now the algorithm queries the oracle for two optimal vertex covers, one for $G\add e_1$ and one for $G\add e_2$. At least one of them has size $k$ (as opposed to $k+1$) and is thus optimal for $G$ as well. 
  If $G$ has the property that, for every pair of adjacent vertices $v_1$ and $v_2$, one of them is universal, then the set of all universal vertices constitutes an optimal vertex cover.
\end{proof}

\section{Full Proof of \texorpdfstring{\cref{thm:TriangleReduction}}{Theorem~\ref{thm:TriangleReduction}}}\label{app:TriangleReduction}

In this appendix, we provide the formal proof of \cref{thm:TriangleReduction}, after restating it for convenience.
\setcounterref{theoremduplicate}{thm:TriangleReduction}
\addtocounter{theoremduplicate}{-1}

\begin{theoremduplicate}
  There is a reduction $g$ from \ThreeSat to \DecVC such that, for every \ThreeCNF-formula $\formula$ and
  for every triangle $T$ in $g(\formula)$, there is a polynomial-time computable
  optimal vertex cover of $g(\formula) - T$.
\end{theoremduplicate} 
\begin{proof}
  Given a \EThreeCNF-formula $\formula$, let $g(\formula) = h(f(\formula))$, where $f$ is the reduction
  from \cref{thm:3sat} and $h$ is the standard reduction from 
  \EThreeSat to \DecVC~\cite{gar-joh:b:int}.
  Let $T$ be a  triangle in $g(\formula)$. Then $T$ corresponds to a clause ${\clause}$ in $f(\formula)$ and
  the satisfying assignments of $f(\formula) - {\clause}$ correspond to optimal vertex covers
  of $g(\formula) - T$, in a polynomially computable way. Hence, we can
  compute in polynomial time an optimal vertex cover of $g(\formula) - T$ from the polynomial-time
  computable satisfying assignment for $f(\formula) - {\clause}$.
\end{proof}

\end{appendix}

\end{document}